\documentclass[reqno,11pt]{amsart}

\usepackage{palatino}

\usepackage{amsmath}
\usepackage{amsfonts,amssymb,amsmath,latexsym,wasysym,mathrsfs,bbm,stmaryrd}
\usepackage[all]{xy}
\usepackage{color}
\usepackage{graphicx,float,wrapfig}

\def\R{\mathbb R}

\def\C{\mathbb C}

\def\Z{\mathbb Z}
\def\N{\mathbb N}

\def\A{\mathscr{A}}

\def\1{\mathbbm 1}

\def\k{\kappa}

\newcommand{\interior}[1]{\raise0.2ex\hbox{$\displaystyle{\mathop{#1}^{\circ}}$}}

\renewcommand\phi{\varphi}
\renewcommand\emptyset{\varnothing}
\newtheorem{theorem}{Theorem}[section]
\newtheorem{prop}{Proposition}[section]

\newtheorem{lem}{Lemma}[section]
\newtheorem*{theorem*}{Theorem}

\newtheorem{proposition}[theorem]{Proposition}
\newtheorem{defn}[theorem]{Definition}

\numberwithin{equation}{section} 
\theoremstyle{remark}
\newtheorem*{remark*}{Remark}
\newtheorem{remark}[theorem]{Remark}
\newcounter{Examplecount}
\theoremstyle{remark}
\long\def\symbolfootnote[#1]#2{\begingroup%
\def\thefootnote{\fnsymbol{footnote}}\footnote[#1]{#2}\endgroup}
\begin{document}

\title{On the norm of the $q$-circular operator}
\author{Natasha Blitvi\'c}

\begin{abstract} 
The \emph{$q$-commutation relations}, formulated in the setting of the \emph{$q$-Fock} space of Bo\.zjeko and Speicher,
interpolate between the classical commutation relations (CCR) and the classical anti-commutation relations (CAR) defined on the classical bosonic and fermionic Fock spaces, respectively. Interpreting the $q$-Fock space as an algebra of ``random variables" exhibiting a specific commutativity structure, one can construct the so-called \emph{$q$-semicircular} and \emph{$q$-circular} operators  
acting as $q$-deformations of the classical Gaussian and complex Gaussian random variables, respectively. While the $q$-semicircular operator is generally well understood, many basic properties of the $q$-circular operator (in particular, a tractable expression for its norm) remain elusive. 

Inspired by the combinatorial approach to free probability, we revist the combinatorial formulations of the $q$-semicircular and $q$-circular operators. We point out that the combinatorics of the $q$-semicircular operator are given by the chord-crossing diagrams developed by Touchard in the 1950s and distilled by Riordan in 1974. This observation leads to a mostly closed-form (viz. finite alternating sum) expression for the $2n$-norm of the $q$-semicircular. Extending these norms as a function in $q$ onto the complex unit ball and taking the $n\to\infty$ limit, we recover the familiar expression for the norm of the $q$-semicircular and show that the convergence is uniform on the compact subsets of the unit ball. In contrast, the $2n$-norms of the $q$-circular operator are given by a restriction of the chord-crossing diagrams to diagrams whose chords are \emph{parity-reversing} (i.e. chords which, labeling the points consecutively, connect even to odd integers only), which have not yet been characterized in the combinatorial literature. We derive certain combinatorial properties of these objects, including closed-form expressions for the number of such diagrams of any size with up to eleven crossings. These properties enable us to conclude that the $2n$-norms of the $q$-circular operator are significantly less well behaved than those of $q$-semicircular operator. 
\end{abstract}

\maketitle

\symbolfootnote[0]{This work was partially supported by NSERC Canada and NSF Grant DMS-0701162.}

\vspace{-0.4in}

\section{Introduction and Background}

An algebraic framework for describing Bose-Einstein and Fermi-Dirac statistics is provided by the bosonic and fermionic 
Fock spaces $\mathcal F_{+}(\mathscr H)$ and $\mathcal F_{-}(\mathscr H)$ over a Hilbert space $\mathscr H$. 
The creation and anihilation operators, $\ell$ and $\ell^\ast$, on these classical Fock spaces satisfy the canonical commutation relations (CCR), describing the interactions of Bosons, or, respectively, the canonical anti-commutation relations (CAR) describing Fermions. Namely,
\begin{align*}&\ell^*(g)\ell(f)-\ell(f)\ell^*(g) = (f, g)_{\mathscr H}\mathbf{1}_{\mathcal F_+ (\mathscr H )} &\quad\quad\text{(CCR)}\\
&\ell^*(g)\ell(f)+\ell(f)\ell^*(g) = (f, g)_{\mathscr H}\mathbf{1}_{\mathcal F_- (\mathscr H )}&\quad\quad\text{(CAR)}\end{align*}
for $f,g\in\mathscr H$. A non-relativistic field theory interpolating between the Bose-Einstein statistics and the Fermi-Dirac statistics can be realized in the setting of the so-called \emph{$q$-Fock space} of Bo\.zejko and Speicher \cite{Bozejko1991}.  Constructed analogously to the classical Fock spaces, which are obtained by completing symmetrized/anti-symmetrized direct sums of single-particle Hilbert spaces via an appropriate inner product, the $q$-Fock space admits the creation and anihilation operators $\ell_q$ and $\ell_q^\ast$ satisfying the \emph{q-commutation relation} ($q$-CR) \cite{Bozejko1991}:\begin{align*}
&\ell_q^*(g)\ell_q(f)-q \ell_q(f)\ell_q^*(g) = (f, g)\mathbf{1}_{\mathcal F_q (\mathscr H )},\quad q\in (-1,1)& \quad\quad \text(q\text{-CR)}.\end{align*}

Interpreting the bounded operators on $\mathcal F_q(\mathscr H)$ as (admittedly non-classical) random variables, the above framework yields a continuum of probability theories admitting various degrees of commutativity. We refer to this interpolation as the \emph{$q$-deformed probability.} Of special interest is the fact that, for $q=0$, the $q$-Fock space construction of \cite{Bozejko1991} yields the full Boltzmann Fock space of free probability. Within the last decade, the diagramatic approach to free probability has yielded both a collection of powerful results and a beautiful theory \cite{NicaSpeicher}, and it is natural to hope that gains may be derived by an analogous approach to $q$-deformed probability.  Adopting the combinatorial view and extending the non-crossing partitions that underlie free probability to partitions with fixed numbers of crossings, we characterize the behavior of the norms of two $q$-commutative random variables that should be interpreted as the Gaussian and complex Gaussian random variables suitably deformed for to the $q$-commutative setting. The surprising consequence is that viewed through this lens, the $q$-deformation of the complex Gaussian random variable turns out to be drastically worse-behaved than the analogously deformed Gaussian.

More concretely, introduced in \cite{Bozejko1991} by Bo\.zejko and Speicher, the $q$-deformation of the Gaussian random variable is termed the \emph{q-Gaussian} or \emph{q-semicircular} distribution. It corresponds to the element $s_q\in \mathscr B(\mathcal F_q(\mathscr H))$ given by $s_q=\ell_q+\ell_q^*$. The $q$-semicircular element recovers the classical Gaussian distribution in the  limit $q\to 1$, whereas the $q=0$ case yields the semi-circular element of Voiculescu \cite{Voiculescu1992}. More generally, Bo\.zejko, K{\"u}mmerer and Speicher \cite{Bozejko1997} succeeded in defining a Brownian motion on this space, while Biane~\cite{Biane1997b}, Mingo \& Nica~\cite{Mingo2001}, and Kemp~\cite{Kemp2005} constructed an analogously deformed circular element, that is, the $q$-deformation of the \emph{complex} Gaussian. The resulting \emph{$q$-circular element} $c_q\in \mathscr B(\mathcal F_q(\mathscr H))$ is expressed as $$c_q=\frac{\ell_1+\ell_1^*+i(\ell_2+\ell_2^*)}{\sqrt 2},$$ where $\ell_1,\ell_2$ are the creation operators corresponding to two orthonormal vectors in $\mathscr H$. 
As expected, considering the $q\to 1$ limit of $c_q$ recovers the complex Gaussian distribution, while letting $q=0$ yields the circular element of free probability, realized as a suitably normalized sum of two free semi-circular elements. However, while the $q$-semicircular distribution is well understood, both the analytical and combinatorial structure of the $q$-circular operator remain in the dark. Our starting point is a fundamental property of the semi-circular element $s_q$, namely the fact \cite{Bozejko1991} that the norm of $s_q$ is given by \begin{equation} ||s_q||=\frac{2}{\sqrt{1-q}}.
\label{sproperties}
\end{equation}

While the analogue of the above expression for the $q$-circular operator remains elusive, we point out and characterize a striking discrepancy in the behavior of the $p$-norms of the two operators, for $p$ an even integer. For $n\in\mathbb N$, the $2n$-\emph{norms} of the $q$-circular and the $q$-semicircular are naturally defined when the two operators are interpreted as elements of a $C^\ast$-probability space (cf. Section~\ref{background}). Focusing on the \emph{moments} of $c_q$ and $s_q$, the $2n$-norms of the two operators
admit a combinatorial definition that relies on the notion of crossings in pairings (perfect matchings) on an ordered set.
Specifically, letting $\{\gamma_n(q)\}$ and $\{\lambda_n(q)\}$ respectively denote the sequences of the $2n$-norms of the $q$-semicircular and $q$-circular operators, the corresponding operator norms are realized as the limits
\begin{equation}||s_q||=\lim_{n\to\infty} \gamma_n(q).\label{convergenceEqs}\end{equation}
\begin{equation}||c_q||=\lim_{n\to\infty} \lambda_n(q).\label{convergenceEq}\end{equation}
Viewed as functions in $q$, $\gamma_n$ and $\lambda_n$ turn out to be given as $1/(2n)^\text{th}$ powers of certain combinatorial generating functions. 
Namely, the generating functions associated with $\gamma_n$ count the sequences of crossings (cf. Section\ref{background}) of pairings on the set $[2n]:=\{1,\ldots,2n\}$, whereas those associated with $\lambda_n$ count crossings in \emph{parity-reversing pairings}, introduced in Section~\ref{background}. 

Counting crossings in parity-reversing pairings is a hard combinatorial problem, akin to problems of Touchard~\cite{Touchard1952,Touchard1950,Touchard1950-2} and Corteel~\cite{Corteel2007}, overviewed in Section~\ref{background}. The former was originally solved by moderately strenuous manipulations involving continued fraction expansions \cite{Touchard1952,Touchard1950,Touchard1950-2,Riordan1975} and actively revisited over the course of the four decades that followed. The original solution was reproduced via a mix of bijection and continued fraction expansion \cite{Read1979}, orthogonal polynomials \cite{Ismail1987}, and several levels of non-obvious bijections \cite{Penaud1995}. The problem of Corteel is posed and solved in \cite{Corteel2007} via a bijection to the previous work on the enumeration of totally positive Grassmann cells by Williams \cite{Williams2005}. 

Without fully solving the enumeration problem at hand, we will derive some basic properties of crossings in the parity-reversing pairings. Using these and similar expressions, we show an unexpected property of the norm of the $q$-circular operator. Namely, it turns out that unlike the case of the $q$-semicircular operator, that the convergence of the norms of the $q$-circular operator in (\ref{convergenceEq}) is not particularly ``nice". Our main results are the following.

\begin{theorem} Let $\gamma_1,\gamma_2,\ldots$ be the sequence of $2n$-norms of the $q$-semicircular operator and consider the sequence of complex-valued functions $\tilde \gamma_1,\tilde \gamma_2,\ldots$ defined on the unit ball $B_{\mathbb C}:= \{z\in\mathbb C, |z|<1\}$ as
$$\tilde \gamma_n(q)=\left(\frac{1}{(1-q)^n}\sum_{k=-n}^n(-1)^kq^{k(k-1)/2}{{2n}\choose{n+k}}\right)^\frac{1}{2n}.$$
Then $\tilde \gamma_n$ analytically extend $\gamma_n$ on $B_{\mathbb C}$. Moreover,
$$\tilde \gamma_n(q)\to\frac{2}{\sqrt{1-q}},\quad q\in B_{\mathbb C}$$
and the convergence is uniform on compact subsets of $B_{\mathbb C}$.
\label{thmgood}
\end{theorem}

\begin{theorem} Let $\lambda_1,\lambda_2,\ldots$ be the sequence of $2n$-norms of the $q$-circular operator. Then, there exists no complex neighborhood of the origin on which $\lambda_n$ have analytic continuations that converge uniformly on compact sets.
\label{thmbad}
\end{theorem}

One concrete consequence of Theorem~\ref{thmbad} is that if $||c_q||$ (as a function in $q$) can be analytically extended onto some complex neighborhood of $(-1,1)$, and thus cast within the rich framework of holomorphic functions, the extension is difficult to achieve via the $2n$ norms. For instance, if $||c_q||$ can be represented by a power series that converges on some neighborhood of the origin, the coefficients of the power series \emph{cannot} be computed as limits of the coefficients in the power series representation of $\lambda_n$. (Specifically, Lemma~\ref{lemma11} will show that some of these coefficients diverge as $n\to\infty$.)
This is surprising, and also a little unfortunate, as the limiting procedure of (\ref{convergenceEq}) otherwise shows great promise in endowing $c_q$ with natural structural insight.

The remainder of this paper is organized as follows. Section~\ref{background} surveys some elements of non-commutative probability and contrasts the definitions and basic properties of the $q$-semicircular and $q$-circular operators. It overviews the combinatorial framework of Touchard that is concerned with crossings in pairings on an ordered set and, realizing that the same combinatorial structure is associated with the norm of the $q$-semicircular element, derives Theorem~\ref{thmgood}. Considering next the norm of the $q$-circular operator, Section~\ref{background} also introduces the combinatorial framework of crossings in parity-reversing pairings and compares it to the related problem of Corteel. Section~\ref{combinatorics} then develops combinatorial properties of crossings on parity-reversing pairings, which are subsequently used in Section~\ref{derivative} to derive the above Theorem~\ref{thmbad}.

\section{Combinatorial structure of $q$-circular and $q$-semicircular operators}
\label{background}

Considering the elements $c_q$ and $s_q$ as elements of a $C^\ast$ \emph{probability space} allows us to equivalently define the two operators via their \emph{moments}. In turn, the corresponding moments admit a clear combinatorial interpretation and, as such, should constitute a starting point of any exploration of the combinatorial underpinnings of $q$-commutative probability. At this point, several definitions are in order.

\begin{defn} A $\ast$-probability space $(\mathcal A, \phi)$ consists of 
\begin{itemize}
\item a unital algebra $\mathcal A$ equipped with an anti-linear $\ast$-operation $a\mapsto a^\ast$ such that $(a^\ast)^\ast=a$ and $(ab)^\ast=b^\ast a^\ast$ for all $a,b\in \mathcal A$;
\item a unital linear functional $\phi:\mathcal A\to \C,\quad \phi(1_\mathcal A)=1$ with the property that for any element $a\in \mathcal A$,
$$\phi(a^\ast a)\geq 0\quad\quad\quad\text{and}\quad\quad\quad\phi(a^\ast a)=0\implies a=0.$$
\end{itemize}

An element $a\in\mathcal A$ is considered to be a \emph{non-commutative random variable}. If $\mathcal A$ is additionally a $C^\ast$ algebra, then $(\mathcal A, \phi)$ is a \emph{$C^\ast$-probability space}.\label{Cast}
\end{defn}

\begin{defn} Given a $\ast$-probability space $(\mathcal A, \phi)$ and a non-commutative random variable $a\in\mathcal A$, the \emph{$\ast$-moments} (or, simply, \emph{moments}) of $a$ refer to complex numbers of the form
$$\phi(a^{\varepsilon(1)}\ldots a^{\varepsilon(n)})$$
where $n\in\mathbb N$ and $\varepsilon(1),\ldots,\varepsilon(n)\in\{1,\ast\}.$\label{defMoments}
\end{defn}

\begin{prop} (e.g.~\cite{NicaSpeicher}) Given a $C^\ast$-probability space $(\mathcal A, \phi)$, for every $a\in \mathcal A$ the norm of $a$ is given by $$||a||=\lim_{n\to\infty}(\phi(a^*a)^n)^{1/2n}.$$ \label{thmNorm}
\end{prop}

The transition from the realm of non-commutative probability theory to that of combinatorics is in the present case accomplished by connecting Defnition~\ref{defMoments} and Proposition~\ref{thmNorm} to the following combinatorial construct.

\begin{defn}
A \emph{pairing} $\pi=\{b_1,\ldots,b_p\}$ on $[n]:=\{1,\ldots,n\}$ is a partition of $[n]$ into blocks of size 2, where $p=n/2$ for $n$ even and $\pi=\emptyset$ otherwise. 

Given a pairing $\pi$ on $[n]$, two blocks $\{a_i,b_i\}$ and $\{a_j,b_j\}$ in $\pi$ are said to \emph{cross} if either $a_i < a_j < b_i < b_j$ or $a_j < a_i < b_j < b_i$. The number of crossings of $\pi$, denoted $\text{cr}(\pi)$, is given by
$$\text{cr}(\pi)=|\{(i,j)\mid 1\leq i<j\leq p, b_i \text{ and } b_j \text{ cross }\}|.$$
\label{defcross}
\end{defn}

A convenient diagramatic representation of pairings and their crossings will be introduced shortly. 
Meanwhile, making use of the previously-defined combinatorial notions of pairings and crossings, the $q$-semicircular and $q$-circular elements can be equivalently defined as follows, cf. \cite{Bozejko1991,Mingo2001}.

\begin{defn}
Let $(\A,\varphi)$ be a $C^\ast$-probability space, and let $q\in(-1,1)$. 

The element $s\in\A$ is said to be a $q$-semicircular element of $\A$ if for every $n\geq 1$, we have $\phi(s^n)=0$ for $n$ odd and
$$\phi(s^{2n})=\sum_{\pi\in\mathcal P(2n)}q^{\text{cr}(\pi)},$$
where $\mathcal P(2n)$ denotes the set of all pairings $\pi=\{\{a_1,b_1\},\ldots,\{a_p,b_p\}\}$ on $[2n]=\{1,\ldots,2n\}$. \cite{Bozejko1991}

The element $c\in\A$ is said to be a $q$-circular element of $\A$ if for every $n\geq 1$ and all  $\varepsilon:[n]\to\{1,\ast\}^n$, we have $\phi(c^{\varepsilon(1)}\ldots c^{\varepsilon(n)})=0$ for $n$ odd, and
$$\phi(c^{\varepsilon(1)}\ldots c^{\varepsilon(2n)})=\sum_{\pi\in\mathcal {P}_{\varepsilon}(2n)}q^{\text{cr}(\pi)} ,$$
where $\mathcal {P}_{\varepsilon}(2n)$ denotes the set of all pairings $\pi=\{\{a_1,b_1\},\ldots,\{a_p,b_p\}\}$ on $[2n]$ with the property that $\varepsilon(a_i)\neq\varepsilon(b_i)$ for all $1\leq i\leq p$. \cite{Mingo2001}
\label{def2}
\end{defn}

In the light of Definition~\ref{def2} and Proposition~\ref{thmNorm}, the norms of the $q$-semicircular and $q$-circular elements of $\A$ are realized as
\begin{equation}||s_q||=\lim_{n\to\infty}\left(\phi((s^*s)^n)\right)^{\frac{1}{2n}}=\lim_{n\to\infty}\left(\sum_{\pi\in\mathcal P(2n)}q^{\text{cr}(\pi)}\right)^{\frac{1}{2n}},\label{norms}\end{equation}
and
\begin{equation}||c_q||=\lim_{n\to\infty}\left(\phi((c^*c)^n)\right)^{\frac{1}{2n}}=\lim_{n\to\infty}\left(\sum_{\pi\in\mathcal{P}_{\varepsilon}(2n)}q^{\text{cr}(\pi)}\right)^{\frac{1}{2n}},\label{normc}\end{equation}
for $\varepsilon:[n]\to\{1,\ast\}^n$ given by
$$\varepsilon(n)=\left\{\begin{array}{ll} 1&n\text{ even}\\\ast&n\text{ odd}\end{array}\right.$$
Since the semi-circular element is self-adjoint, $\phi((s^*s)^n)=\phi(s^{2n})$ in (\ref{norms}) and the sum is indexed over all pairings on $[2n]$. However, $c$ is not self-adjoint, and the key observation at this point is that ${P}_{\varepsilon}(2n)$, which indexes the sum in (\ref{normc}), is the set of pairings on $[2n]$ whose every pair contains an even and an odd integer, thus motivating the following definition.

\begin{defn}
A \emph{parity-reversing pairing} on $[2n]$ is a pairing $\pi=\{\{a_1,b_1\},\ldots,\{a_n,b_n\}\}$ on $[2n]$ with the property that for all $i=1\ldots,n$, $\{a_i,b_i\}$ contains one even and one odd element. Let $\mathcal R_{n}$ denote the collection of parity-reversing pairings on $[2n]$, and let $\mathcal R_{n,k}$ denote the subset of those pairings containing exactly $k$ crossings (cf. Definition~\ref{defcross}).
\label{defpairings}
\end{defn}

\subsection{Combinatorics of the $q$-semicircular operator}
In order to formulate the expression (\ref{norms}) combinatorially, let $t_{n,k}$ count the number of pairings $\pi$ on  $[2n]$ admitting exactly $k$ crossings and consider the corresponding generating function $T_n:\C \to \C$ given by $T_{n}(q)=\sum_{k}t_{n,k}q^k$. The norm of the $q$-semicircular element is then realized as the following limit:
\begin{equation}||s_q||=\lim_{n\to\infty} (T_n(q))^\frac{1}{2n}.\label{normTouchard}\end{equation}
Note that the polynomials $T_{n}(q)$ are the moments of the $q$-Hermite polynomials of \cite{Ismail1987} arising from the Askey-Wilson integration scheme. This should not come as a surprise, as the $q$-Hermite polynomials form the orthogonal polynomial sequence associated with the $q$-semicircular operator \cite{Bozejko1997,Kemp2005}. 

According to Riordan \cite{Riordan1975}, the problem of counting the non-crossing pairings on $[2n]$ was first posed and solved by Alfred Errera \cite{Errera1931}. The result is given by $C_n={{2n}\choose n}/(n+1)$ and is one of the many occurrences of Catalan numbers in enumerative combinatorics. 
To count the total number of pairings (summed over all crossings), the reader may easily verify that the set of all pairings on $[2n]$ is in bijective correspondence with the fixed-point-free involutions on $[2n]$. It follows that there is a total of $(2n-1)(2n-3)\ldots 5\cdot 3\cdot 1=(2n)!!$ pairings on $[2n]$, exactly $n!$ of which are are parity-reversing.

To deal with the general number of crossings in pairings, it is beneficial to diagramatically representing the sequence $t_{n,k}$ via  \emph{chord-crossing diagrams}, introduced by Touchard \cite{Touchard1952,Touchard1950,Touchard1950-2} in the context of the postage stamp problem. Starting with a pairing $\pi\in \mathcal P(2n)$, the corresponding chord-crossing diagram is obtained by representing the set $[2n]$ as points on a circle and the pairs in $\pi$ as connecting disjoint chords (i.e. with no two cords sharing an endpoint). The notion of a pairing therby gains structural insight by which Definition~\ref{defcross} of a crossing becomes natural. Three pairings on $\{1,\ldots,10\}$ with five, three, and zero crossings, respectively, are represented as chord-crossing diagrams in Figure~\ref{chordsFig}.

\begin{figure}[tpb]\includegraphics[scale=0.5]{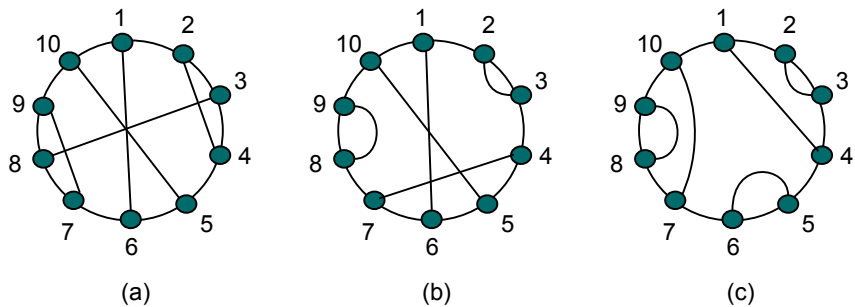}\caption{Chord-crossing diagrams representing the pairings: \mbox{(a) $\{\{1,6\},\{2,4\},\{3,8\},\{5,10\},\{7,9\}\}$,} (b) $\{\{1,6\},\{2,3\},\{4,7\},\{5,10\},\{8,9\}\}$, {(c) $\{\{1,4\},\{2,3\},\{5,6\},\{7,10\},\{8,10\}\}$.}}\label{chordsFig}\end{figure}

The chord-crossing diagrams were studied in much depth by Touchard, Riordan, and others \cite{Touchard1952,Touchard1950,Touchard1950-2,Riordan1975,Penaud1995,Flajolet2000}. A general expression for the number of pairings with $k$ crossings, where $k$ takes values in $\{1,\ldots,{n\choose 2}\}$ was implicit in the work of Touchard \cite{Touchard1952,Touchard1950,Touchard1950-2} and distilled by Riordan \cite{Riordan1975}.
The formula of Touchard and Riordan for counting the number of chord-crossing diagrams with a fixed number of crossings is as follows. The formula was originally derived via generating functions, with a bijective proof later provided by Penaud \cite{Penaud1995}.
\begin{theorem}\cite{Riordan1975,Touchard1952} For $n\in\N$,
\begin{equation}T_{n}(q)=\frac{1}{(1-q)^n}\sum_{k=-n}^n(-1)^kq^{k(k-1)/2}{{2n}\choose{n+k}}.\label{Touchard}\end{equation}\label{thmTouchard}
\end{theorem}

Combining (\ref{norms}) and (\ref{Touchard}), the norm of the $q$-semicircular operator is expressed via the $\{T_n(q)\}_{n\in\N}$ sequence of generating functions
$$||s_q||=\lim_{n\to\infty} (T_n(q))^\frac{1}{2n}.$$
Evaluating this limit recovers the familiar expression (\ref{sproperties}) and yields the following proof of Theorem~\ref{thmgood}.

{\it Proof of  Theorem~\ref{thmgood}.}
Consider the sequence $\tilde\gamma_n$ of complex-valued functions  on the complex unit ball $B_{\mathbb C}=\{z\in \mathbb C; |z|<1\}$ given by
$$\tilde\gamma_n(q)=\left|\frac{1}{(1-q)^n}\sum_{k=-n}^n(-1)^kq^{k(k-1)/2}{{2n}\choose{n+k}}\right|\,e^{i\phi_{n}(q)},$$
where $\phi_{n}(q)\in [-\pi,\pi)$ is the phase of $(1-q)^{-n}\sum_{k=-n}^n(-1)^kq^{k(k-1)/2}{{2n}\choose{n+k}}$. 
We will first show that $\tilde\gamma_n(q)^\frac{1}{2n}\to 2/\sqrt{1-q}$ for all $q\in B_{\mathbb C}$ and subsequently argue that $\tilde\gamma_n$ is analytic on $B_{\mathbb C}$ for all $n\in\mathbb N$ and that the convergence is uniform on compact subsets of $B_{\mathbb C}$.

Given any $q\in B_{\mathbb C}$, write
\begin{eqnarray*}
\sum_{k=-n}^n(-1)^kq^{k(k-1)/2}{{2n}\choose{n+k}}&=&(-1)^{n+1}\sum_{k=1}^{2n+1}(-1)^kq^{(k-n-1)(k-n-2)/2}{{2n}\choose{k-1}}.
\end{eqnarray*}
Recalling that $|q|<1$ and noting that ${{2n}\choose{n+k}}$ is maximized at $k=0$ yields the crude upper bound
\begin{equation}\left|\sum_{k=-n}^n(-1)^kq^{k(k-1)/2}{{2n}\choose{n+k}}\right|\leq (2n+1){2n\choose n}\label{ineq1a}\end{equation}
A lower bound requires a more careful handling of the alternating sum. For this, note that the magnitude of $q^{(k-n-1)(k-n-2)/2}$ is maximized at $k\in\{n+1,n+2\}$, in which case $q^{(n+1-n-1)(n+1-n-2)/2}$\, $=1$.
Moreover, as ${{2n}\choose{k-1}}$ increases with $k$ on $\{1,\ldots,n+1\}$ and decreases on $\{n+2,\ldots,2n+1\}$, one has that $q^{(k-n-1)(k-n-2)/2}{{2n}\choose{k-1}}$ is increasing in $k$ for  $k=1,\ldots,n+1$ and decreasing in $k$ for $k=n+2,\ldots,2n+1$. Considering first the terms corresponding to $k=1,\ldots,n+1$, it follows that
\begin{equation}\left|\sum_{k=1}^{n+1}(-1)^kq^{(k-n-1)(k-n-2)/2}{{2n}\choose{k-1}}\right | \geq {{2n}\choose{n}}-\sum_{k=1}^{n}|q|^{(k-n-1)(k-n-2)/2}{{2n}\choose{k-1}}.\label{ineq1aa}\end{equation}
But,
$$\sum_{k=1}^{n}|q|^{(k-n-1)(k-n-2)/2}{{2n}\choose{k-1}}\leq \sum_{k=1}^{n}|q|^k{{2n}\choose{k-1}}\leq \sum_{k=1}^{n}|q|^k{{2n}\choose{k}}\leq  (1+|q|)^{2n}.$$
Thus,
\begin{equation}\left|\sum_{k=1}^{2n+1}(-1)^kq^{(k-n-1)(k-n-2)/2}{{2n}\choose{k-1}}\right|\geq {{2n}\choose{n}}-(1+|q|)^{2n}.\label{ineq1b}\end{equation}
Next, for $k=n+2,\ldots,2n+1$,
$$\left|\sum_{k=n+2}^{2n+1}(-1)^kq^{(k-n-1)(k-n-2)/2}{{2n}\choose{k-1}}\right |\leq {{2n}\choose{n+1}}+\sum_{k=n+3}^{2n+1}|q|^{(k-n-1)(k-n-2)/2}{{2n}\choose{k-1}},$$
and
$$\sum_{k=n+3}^{2n+1}|q|^{(k-n-1)(k-n-2)/2}{{2n}\choose{k-1}}\leq \sum_{k=n+3}^{2n+1} |q|^{k}{{2n}\choose{k-1}} \leq (1+|q|)^{2n+1}.$$ Thus,
\begin{equation}\left|\sum_{k=n+2}^{2n+1}(-1)^kq^{(k-n-1)(k-n-2)/2}{{2n}\choose{k-1}}\right|\leq {{2n}\choose{n+1}}+(1+|q|)^{2n+1}.\label{ineq1c}\end{equation}
Letting $a_n(q):=\sum_{k=1}^{n+1}(-1)^kq^{(k-n-1)(k-n-2)/2}{{2n}\choose{k-1}}$ and $b_n(q):=\sum_{k=n+2}^{2n+1}(-1)^kq^{(k-n-1)(k-n-2)/2}{{2n}\choose{k-1}}$, write
$$|a_n(q)|-\left|b_n(q)\right|\leq\left|\sum_{k=1}^{2n+1}(-1)^kq^{(k-n-1)(k-n-2)/2}{{2n}\choose{k-1}}\right|=|a_n(q)+b_n(q)|\leq \left|a_n(q)\right|+ \left|b_n(q)\right|.$$
Applying the inequalities (\ref{ineq1aa})-(\ref{ineq1c}) then yields,
\begin{equation}{{2n}\choose{n}}-{{2n}\choose{n+1}}-2(1+|q|)^{2n+1}\leq \left|\sum_{k=1}^{2n+1}(-1)^kq^{(k-n-1)(k-n-2)/2}{{2n}\choose{k-1}}\right|\leq (2n+1){{2n}\choose{n}}.\label{ineq1e}\end{equation}
Noticing that ${{2n}\choose{n}}-{{2n}\choose{n+1}}=\frac{1}{n+1}{2n\choose n}$, ${{2n}\choose n}\sim \frac{4^n}{n^{1/2}\sqrt{\pi}}$ and $2(1+|q|)^{2n+1}< 2^{2n+2}$, one immediately obtains that
$${{2n}\choose{n}}-{{2n}\choose{n+1}}-2(1+|q|)^{2n+1}\sim \frac{4^n}{n^{3/2}\sqrt{\pi}}$$ and
$${{2n}\choose{n}}(n+1)+{{2n}\choose{n+1}}+(1+|q|)^{2n+1}\sim n^{1/2} \frac{4^n}{\sqrt{\pi}}.$$
It follows that
\begin{equation}\left|\sum_{k=1}^{2n+1}(-1)^kq^{(k-n-1)(k-n-2)/2}{{2n}\choose{k-1}}\right|^{\frac{1}{2n}}\to 2\quad\quad\text{as }\,n\to\infty.\label{mainclaim}\end{equation}
Now recall that
$$\tilde\gamma_n(q)^\frac{1}{2n}=\left|\frac{1}{(1-q)^n}\sum_{k=1}^{2n+1}(-1)^kq^{(k-n-1)(k-n-2)/2}{{2n}\choose{k-1}}\right|^{\frac{1}{2n}}\,e^{i\frac{\phi_n(q)}{2n}},$$
where $\phi_n(q)\in [-\pi,\pi)$. Thus, by (\ref{mainclaim}) and the fact that  $e^{i\frac{\phi_n(q)}{2n}}\to 1$, we finally obtain that
\begin{equation}\lim_{n\to\infty}\tilde\gamma_n(q)^{\frac{1}{2n}}=\frac{2}{\sqrt{1-q}}\label{eqconvergence1}\end{equation} for all $q\in B_{\mathbb C}$. 

At this point, it is easy to see that for all $n\in\mathbb N$, $\tilde \gamma_n$ is analytic on $B_{\mathbb C}$. Specifically, $(1-q)^{-\frac{1}{2}}$ as a function in $q$ is analytic on $B_{\mathbb C}$ and $\sum_{k=-n}^n(-1)^kq^{k(k-1)/2}{{2n}\choose{n+k}}$ is a polynomial which, by the lower bound in (\ref{ineq1e}), is nowhere vanishing on $B_{\mathbb C}$.

Finally, to show that the convergence in (\ref{eqconvergence1}) is uniform on compact subsets of the unit ball $B_{\mathbb C}$, fix $\epsilon>0$ and let $K$ denote some such subset. Let $t_n(q):=\sum_{k=-n}^n(-1)^kq^{k(k-1)/2}{{2n}\choose{n+k}}$ and note that it suffices to show that $t_n^{1/(2n)}\to 2$ uniformly on $K$. 
Fix $\epsilon>0$ and write
$$\bigg|2-t_n(q)^{1/(2n)}\bigg|\leq \bigg|2-|t_n(q)|^{1/(2n)}\,\bigg|+ \bigg|t_n(q)^{1/(2n)}-|t_n(q)^{1/(2n)}|\,\bigg|\,.$$
By (\ref{ineq1e}),
$$\bigg|2-|t_n(q)|^{1/(2n)}\bigg|\leq\max\left\{\left|2-\left((2n+1){{2n}\choose{n}}\right)^{1/(2n)}\right|,\sup_{q\in K}\left|2-\left(\frac{1}{n+1}{{2n}\choose{n}}-2(1+|q|)^{2n+1}\right)^{1/(2n)}\right|\right\}.$$
Since $K$ is compact, the above supremum is achieved at some $q_*$, where $|q_\ast|<1$. As previously shown,
$$\left((2n+1){{2n}\choose{n}}\right)^{1/(2n)}\to 2 \quad\quad\text{and}\quad\quad\left(\frac{1}{n+1}{{2n}\choose{n}}-2(1+|q_*|)^{2n+1}\right)^{1/(2n)}\to 2.$$
Since neither of the above limits depends on $q$, it follows that there exists some integer $m\in\mathbb N$ so that $|2-|t_n(q)|^{1/(2n)}|\leq \epsilon$ for all $n\geq m$ and all $q\in K$.
Additionally, for all $n\geq m$ and all $q\in K$,
$$\bigg|t_n(q)^{1/(2n)}-|t_n(q)^{1/(2n)}|\,\bigg|\leq(2+\epsilon)\,\bigg|e^{\frac{\angle{t_n(q)}}{2n}}-1\bigg|\leq (2+\epsilon)\,\bigg|e^{2\pi/(2n)}-1\bigg|\leq (2+\epsilon)\frac{2\pi}{n},$$
where $|e^{2\pi/(2n)}-1|\leq 2\pi/n$ for all $n$ large enough (where the threshold does not depend on $q$).
Thus, for any $\epsilon>0$, there exists some $m\in\mathbb N$ so that
$$\bigg|2-\left(\sum_{k=-n}^n(-1)^kq^{k(k-1)/2}{{2n}\choose{n+k}}\right)^\frac{1}{2n}\bigg|=\bigg|2-|t_n(q)|^{1/(2n)}\bigg|<\epsilon$$
for all $n>m$ and $q\in K$. In other words, $\left(\sum_{k=-n}^n(-1)^kq^{k(k-1)/2}{{2n}\choose{n+k}}\right)^\frac{1}{2n}\to 2$ uniformly on $K$ and the result follows.
$\hfill\qed$

\subsection{Combinatorics of the $q$-circular operator}
In order to formulate the expression (\ref{normc}) for the norm of the $q$-circular operator combinatorially, let $r_{n,k}=|\mathcal R_{n,k}|$ count the number of parity-reversing pairings $\pi$ on $[2n]$ admitting exactly $k$ crossings and let $R_{n}(q)=\sum_{k}r_{n,k}q^k$ denote the corresponding generating function. Then,
\begin{equation}||c_q||=\lim_{n\to\infty} (R_n(q))^\frac{1}{2n}.\label{normPRP}\end{equation}

The pairings in $\mathcal R_n$ are likewise realized as a chord-crossing diagrams admitting only those chords which connect even to odd integers. Revisiting Figure~\ref{chordsFig}, note that the pairings in (b) and (c) are parity-reversing, whereas the one in (a) is not. However, in a manner that more closely ties to the original motivation (\ref{normc}), it is worthwhile modifying the chord-crossing diagram when representing  parity-reversing pairings. Specifically, the parity-reversing pairings will from now on be represented directly as $\circ\leftrightarrow\ast$ matchings on $(\circ,\ast)^n$, with two types of elements ($\ast $ and $\circ$) labeled $1$ through $n$ and with chords connecting elements of opposite types. For convenience, let us adopt the convention that $\circ_i$ connects to $\ast_{\sigma(j)}$, where $\sigma$ is the unique permutation corresponding to the given parity-reversing pairing determined using the procedure outlined in the proof of Proposition~\ref{propPermutations}. Figure~\ref{figurePRP} diagramatically represents a parity-reversing pairing on $\{1,\ldots,10\}$ as both a chord-crossing diagram and a $\circ\leftrightarrow \ast$ matching diagram on $(\circ,\ast)^n$. 

\begin{figure}[h]
\includegraphics[scale=0.5]{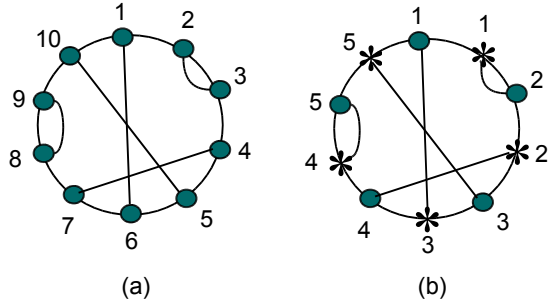}\centering
\caption{Two representations of a parity-reversing pairing on $\{1,\ldots,10\}$:  a) The chord crossing diagram, b) The $\circ\leftrightarrow\ast$ matching on $(\circ,\ast)^5$. Following our convention, the pairing $\{\{1,6\},\{2,3\},\{4,7\},\{5,10\},\{8,9\}\}$ is equivalent to the permutation $\sigma=31524$ and the matching is represented as $\circ_i\mapsto \ast_{\sigma(i)}$.}
\label{figurePRP}
\end{figure}

\begin{remark}
Representing a parity-reversing pairing on $[2n]$ as an $\ast\leftrightarrow \circ$ pairing on $(\ast\circ)^n$, Definition~\ref{defcross} can be reformulated as follows: pairs $(\ast_i,\circ_{\sigma(i)})$ and $(\ast_j,\circ_{\sigma(j)})$, cross if and only if exactly one element from the set $\{j,\sigma(j)\}$ belongs to the interval in $\mathbb N$ whose endpoints are $i$ and $\sigma(i)$.
\label{remarkcross}
\end{remark}

Representing parity-reversing pairings as $\circ\leftrightarrow\ast$ matchings on $(\circ,\ast)^n$ renders clear the following important fact.

\begin{prop} The set of all parity-reversing pairings on $[2n]$, $\mathcal R_n$, is in bijective correspondence with $S_n$, the group of permutations on $[n]$.\label{propPermutations}
\end{prop}
\begin{proof}
Let $\pi$ be a parity-reversing pairing and construct the corresponding permutation $\sigma_\pi$ as follows. Represent the blocks of $\pi$ as ordered pairs $(a_i,b_i)$ where $a_i$ is odd and $b_i$ even for $i=1,\ldots,n$. Order the pairs increasingly along the first coordinate. The result can be written as $(1,e_1),(2,e_2),\ldots,(n,e_n)$, where $\{e_1,\ldots,e_n\} = \{2,4,\ldots,2n\}$. The claim follows by letting $\sigma(i)=e_i/2$.
\end{proof}

A remarkable consequence of Proposition~\ref{propPermutations} is the fact that, as parity-reversing pairings can be essentially thought of as permutations, the associated concept of a crossing provides a natural (diagramatic) notion of \emph{permutation crossings} that is substantially different than the various other previously-proposed notions. Algorithmic applications have lead to a nearly-ubiquitous definition of permutation crossings as derangements of a permutation (e.g. \cite{Biedl2006}). However, the recent re-definition by Corteel \cite{Corteel2007}, shown to exhibit deep connections to other permutation statistics, statistical mechanics, and orthogonal polynomials, is in fact intimately related to the concept of crossings in parity-reversing pairings. 
The crossings introduced by Corteel, referred to as \emph{directed crossings} in the present context, are defined as follows. Consider a permutation $\sigma$ on $[n]$ and a pair $(i,\sigma(i))$. Consider the pair to have \emph{positive orientation} if $i\leq \sigma(i)$ and \emph{negative orientation} if $i>\sigma(i)$. Given a permutation $\sigma$ on $[n]$ and $i\in[n]$ such that $(i,\sigma(i))$ has positive orientation, let $C_+(i)$ denote the set of remaining positively oriented pairs that cross $(i,\sigma(i))$ and, for convenience, let $C_-(i)=\emptyset$. Analogously, given a pair $(j,\sigma(j))$ with negative orientation, let $C_-(i)$ denote the set of remaining negatively oriented pairs that cross $(j,\sigma(j))$ and let $C_+(j)=\emptyset$. More succintly, for a given permutation $\sigma$ on $[n]$ and $i\in[n]$,
$$C_+(i) = \{j\in[n] \mid j <i \leq \sigma(j)<\sigma(i)\},$$
$$C_-(i) = \{j\in[n] \mid j >i>\sigma(j)>\sigma(i)\}.$$
Let the \emph{directed crossing} (our terminology) refer to a crossing that is engendered by two pairs with the same orientation and write $\widetilde{\text{Cr}}(\sigma)$ for the number of directed crossings in a permutation $\sigma$. Then,
$$\widetilde{\text{Cr}}(\sigma):= \sum_{i=1}^n|C_+(i)|+\sum_{i=1}^n|C_-(i)|.$$
For example, revisiting the pairing of Figure~\ref{figurePRP} and recalling the convention $\circ_i\mapsto\ast_{\sigma(j)}$, the pairs $(i,\sigma(i))$ given by $(1,3)$ and $(3,5)$ have positive orientation, while the pairs $(2,1),(4,2),(5,4)$ have negative orientation.

Directed crossings form a subset of the crossings introduced in Definition~\ref{defcross}. For example, revisiting the pairing of Figure~\ref{chordsFig}-b), pairs $(1,6)$ and $(2,4)$ have a positive orientation and pairs $(2,1)$, $(4,2)$, $(5,4)$ have negative orientation. It follows that, though the pairing has a total of three crossings, the only directed crossing is that of $(1,3)$ with $(3,5)$.

The directed crossings in permutations turn out to be specializations of a much broader class of objects, namely the \emph{staircase tableaux} introduced by Corteel and Williams \cite{Corteel2010}. The staircase tableaux are a five-parameter class of labeled Young tableaux of shape $n,n-1,\ldots,1$, shown in  \cite{Corteel2010} to describe the combinatorial underpinnings of both the asymmetric exclusion process (ASEP) and the Askey-Wilson polynomials. In this framework, the generating function associated with the directed permutation crossings is given by the $\alpha=\beta=1$, $\gamma=\delta=0$, $y=1$ specialization of the partition function of the ASEP, previously computed in \cite{Williams2005}, yielding the following expression.

\begin{theorem}\cite{Williams2005,Corteel2007} For $n\in\mathbb N$,
$$\sum_{\sigma\in S_{n}}q^{\widetilde{\text{Cr}}(\sigma)}=\sum_{k=0}^{n}\sum_{i=0}^{k-1}(-1)^i[k-i]_q^{n}q^{k(i-k)}\left({{n}\choose i}q^{k-i}+{{n}\choose{i-1}}\right),$$
where $S_{n}$ denotes the set of all permutations on $[n]$ and $[\ell]_q=1+q+q^2+\ldots+q^{\ell-1}$ for any non-negative integer $\ell$. 
\label{thmCorteel}\end{theorem}

\begin{remark} Since directed crossings form a subset of all crossings, $\text{Cr}(\sigma)\geq \widetilde{\text{Cr}}(\sigma)$ and so
$\sum_\sigma q^{\text{Cr}(\sigma)}\leq \sum_\sigma q^{\widetilde{\text{Cr}}(\sigma)}$ for $q\in [0,1)$. Proceeding along broadly similar lines to the proof of Theorem~\ref{thmgood}, one obtans that for $|q|<1$,
$$\lim_{n\to\infty}\left(\sum_{\sigma\in S_{2n}}q^{\widetilde{\text{Cr}}(\sigma)}\right)^\frac{1}{2n}=\frac{2}{\sqrt{1-q}},$$
yielding the following upper bound for the norm of the $q$-circular operator:
$$ ||c_q||\leq \frac{2}{\sqrt{1-q}},\quad\quad\quad q\in [0,1).$$
The alert reader may in fact note that $\text{Cr}(\sigma)\geq \widetilde{\text{Cr}}(\sigma)+2$, where the bound is met with equality for $\sigma(i)=i+1\!\mod n$, but the distinction disappears when taking the $n\to\infty$ limit of the $1/(2n)$ power of the corresponding sum. However, the above upper can be derived much more directly by recalling the combinatorial definitions of the norms of $s_q$ and $c_q$ in (\ref{norms}) and (\ref{normc}) respectively and noticing that (\ref{normc}) merely sums over fewer diagrams than (\ref{norms}). Thus, for $q\in [0,1)$, $||c_q||\leq ||s_q||$ and the result follows recalling that, by (\ref{sproperties}), $||s_q||=2/\sqrt{1-q}$. In passing, also note that the above constant is the best available for an upper bound of type $1/\sqrt{1-q}$, which stems from the $q=0$ case and the fact that $||c_0||=2$.

A similar, though somewhat weaker, bound valid for all $q\in (-1,1)$ can be derived without recourse to combinatorics. Namely, recalling the formulation of $s_q$ and $c_q$ in terms of the creation and anihilation operators $\ell_i,\ell_i^*$ on the $q$-Fock space, an application of the triangle inequality yields $||c_q||\leq \sqrt 2 ||s_q||= 2\sqrt{2} /\sqrt{1-q}$ for $q\in (-1,1)$

\label{remarkCorteel}\end{remark}

At this time, there is no analogue of Theorems~\ref{thmTouchard} and \ref{thmCorteel} for the generating function restricted to crossings in parity-reversing pairings. However, by deriving basic combinatorial properties of parity-reversing pairings, we will show the surprising fact that a tractable expression for the generating function is in fact needed, as the expression for $||c_q||$ cannot be derived from the asymptotics of the indivial terms of the sequence.

\section{Combinatorics of Parity-reversing Pairings}
\label{combinatorics}
By (\ref{normc}) of the previous section, the norm of the $q$-circular operator is realized as the large-$n$ limit of the generating sequence $R_n(q)=\sum_{k}r_{n,k} q^k$ raised to the power $(2n)^{-1}$. The previous developments in the combinatorics of chord crossing diagrams suggest that obtaining a general expression for $R_n(q)$ may turn out to be relatively difficult and that the resulting expression may not be of a sufficiently tractable form to enable the computation of the limit.  Nevertheless, the basic properties of the doubly-indexed $\{r_{n,k}\}$ sequence and the closed-form expressions for the fixed-$k$ sequences $\{r_{n,k}\}_{n\geq 1}$ for $k\in\{0,1,\ldots,11\}$, derived in the present section, will provide analytic insight into the nature of the convergence taking place. 

\subsection{Connected Diagrams and Associated Transforms}
\label{connectedSec}

The combinatorial technique of much use in diagram enumeration consists of identifying families of naturally occurring irreducibles into which objects can be
decomposed and from which they can be uniquely reconstructed. 
In the present context, the irreducibles considered are termed \emph{connected diagrams} and defined as chord-crossing diagrams whose set of chords cannot be partitioned so that no chord from one part intersects a chord from another.
For example, the diagram of Fig.~\ref{connected diagram}-(a) is connected while that of Fig.~\ref{connected diagram}-(b) is not. The following proposition shows that the connected components of any pairing on $(\ast\circ)^n$ induce noncrossing partitions of  $(\ast\circ)^n$ into blocks of even length and that, conversely, a choice of non-crossing partition and connected components uniquely determines a pairing.

\begin{figure}
\includegraphics[scale=0.5]{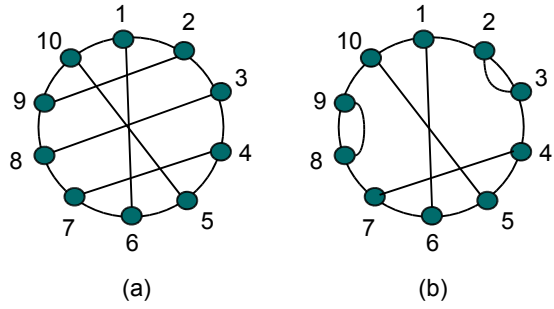}\centering
\caption{Two parity-reversing pairings on $\{1,\ldots,10\}$ represented as chord-crossing diagrams: (a) is connected whereas (b) decomposes into three connected components.}
\label{connected diagram}
\end{figure}

\begin{prop}
Let $b_n$ denote the number of connected diagrams among the parity reversing pairings on $(\ast\circ)^n$ and set $b_0=1$. Then, for all $n\in\mathbb N$,
\begin{equation}\displaystyle n!=\sum_{\substack{\pi\in NC_\text{even}(n)\\\pi=\{V_1,\ldots,V_{|\pi|}\}}} b_{|V_1|/2}\ldots b_{|V_{|\pi|}|/2},\label{decompositionEq}\end{equation}
where $NC_\text{even}(n)$ contains the non-crossing partitions of $[n]$ into blocks of even length.
\end{prop}
\begin{proof}
Consider a parity-reversing pairing $\pi$ on $[2n]$. If $\pi$ is connected, let the corresponding partition of $[2n]$ be the maximum partition (i.e. the set $[2n]$ itself). Otherwise, one can write $\pi=C_1\cup C_2$ where no pair in $C_1$ crosses a pair in $C_2$. In that case, the sets $S_1$ and $S_2$ given by unions of elements in $C_1$ and $C_2$, respectively, form a partition of $[2n]$. If the two parts crossed, there would have to be elements $x_1,y_1\in S_1$ and $x_2,y_2\in S_2$ such that $x_1<x_2<y_1<y_2$ or $x_2<x_1<y_2<y_1$. Either case would imply that the chord $\{x_1,y_1\}$ in $C_1$ crosses the chord in $\{x_2,y_2\}\in C_2$, which is an impossibility.
The resulting partition is therefore non-crossing. Finally, notice that $C_1$ and $C_2$ can be identified in a natural manner with parity-reversing pairings on $\{1,\ldots,|S_1|\}$ and $\{1,\ldots,|S_2|\}$, respectively. Thus, any parity-reversing pairing on $[2n]$ can be decomposed inductively into a non-crossing partition $P_1\cup\ldots \cup P_\ell=[2n]$ where each part $P_i$ corresponds to a connected parity-reversing pairing on $\{1,\ldots,|P_i|\}$. 

Now notice that the above decomposition yields the same result regardless of the choices of $C_1$ and $C_2$ made at every step leads and that, furthermore, the result is unique in the sense that no two distinct pairings can decompose into the same assignment of partitions and connected components. The parity-reversing pairings on $[2n]$ are therefore in bijective correspondence with the collection of non-crossing partitions on $[2n]$ whose each part is assigned a connected parity-reversing pairing compatible with the size of the part.
Transcribing the decomposition into the enumerative language then yields the desired formula.
\end{proof}

In the light of the above decomposition, the diagram of Figure~\ref{connected diagram} (a) uniquely decomposes into a single connected diagram, namely $\{\{1,6\},\{2,9\},\{3,8\},\{4,7\},\{5,10\}\}$, and the associated partition of $\{1,\ldots, 10\}$ is the maximal one (i.e. the set itself). Similarly, the pairing depicted in (b) uniquely decomposes into three connected diagrams, namely $\{\{1,6\},\{4,7\},\{5,10\}\}$, $\{\{2,3\}\}$, and $\{\{5,10\}\}$, corresponding to the parts $\{1,4,5,6,7,10\}$, $\{2,3\}$, and $\{8,9\}$, respectively.

Since connected diagrams are a subset of chord-crossing diagrams on $[2n]$, the notion of crossings extends naturally to connected components. Recall that $r_{n,k}$ denotes the number of parity-reversing pairings on $(\ast\circ)^n$ with $k$ crossings and let
$b_{n,k}$ count the subset of those pairings that are connected. Refining (\ref{decompositionEq}) to take into account the numbers of crossings in a pairing yields the following relation, which provides the starting point for the developments of the next section.

\begin{equation}r_{n,k}=\sum_{\substack{\pi=\{V_1,\ldots,V_{\ell}\}\in NC_{\text{even}}(2n)\\k=k_1+\ldots+k_{\ell}}} b_{|V_1|/2,k_1}\ldots b_{|V_\ell|/2,k_\ell}\label{decompositionEq2}\end{equation}

It is worthwhile pausing here to remark that the above relations admit a particularly concise representation in the form of generating functions. Consider the formal sums $R(z,q)=\sum_{n,k\geq 0} r_{n,k}z^n q^k$, $B(z,q)=\sum_{n,k\geq 0} b_{n,k}z^n q^k$, $B(z)=\sum_{n=0}^\infty b_n z^n$ and $F(z)=\sum_{n=0}^\infty  n!z^n$, referred to as the (ordinary) power series generating functions of the $\{r_{n,k}\}_{n,k\geq 0},\,\{b_{n,k}\}_{n,k\geq 0},\,\{b_n\}_{n\geq}$ and $\{n!\}_{n\geq 0}$ sequences, respectively. We then have the following proposition.

\begin{prop} $F(z)= B(zF^2(z))$ and $R(z,q)=B(zR^2(z,q),q)$.\end{prop}
\begin{proof}
Suppose that the connected diagram containing $\ast_1$, referred to as the root component, contains $i$ chords. The chords of the root component partition the remaining unpaired elements into $2i$ (potentially empty) intervals. Each interval of length $j$ contains $j$ $\circ$-symbols and $j$ $\ast$-symbols, so there are $j!$ possible pairings on each interval such that no pair crosses the root component. The set of all pairings on $(\ast\circ)^n$ is in bijective correspondence with the choice of (1) the root component and (2) the pairings on the resulting intervals. It follows that 
$$n!=\sum_{i=0}^n b_i\sum_{\substack{a_1,\ldots,a_{2i}\geq 0\\a_1+\ldots+a_{2i}=n-i}}a_1!\ldots a_{2i}!.$$
At this point, the reader may readily verify that the right-hand side of the above expression corresponds to the coefficient of $z^n$ in $B(zF^2(z))$. 

The proof of the second relation proceeds analogously taking into account the crossings numbers.
\end{proof}

\subsection{Crossings in Parity-reversing Chord Diagrams} We next derive several elementary properties of the (doubly-indexed) triangular sequences $\{r_{n,k}\}_{n\geq 1,k\geq 0}$ and $\{b_{n,k}\}_{n\geq 1,k\geq 0}$ which are used to derive closed-form expressions for $\{r_{n,k}\}_{n\geq 1}$ for fixed $k\leq 11$. In terms of the structural insight they provide, Proposition~\ref{propnocross} presents an immediate similarity and two interesting differences between (unrestricted) pairings and the parity-reversing pairings. Namely, it turns out that while the parity-reversing restriction leaves unchanged the collection of non-crossing pairings, it eliminates all pairings with exactly one or two crossings. More generally, Proposition~\ref{prop pairings} considers the extremes of the crossing numbers in parity-reversing pairings and, in particular, in connected parity-reversing pairings. These extremes are illustrated in Figures~\ref{min pairings} and \ref{max pairings} for $n=3,4,5$.

\begin{proposition} For all $n\in\mathbb N$,
\begin{enumerate}
\item $\displaystyle r_{n,0}=\frac{1}{n+1}{{2n}\choose n}$
\item $\displaystyle r_{n,1}=0$
\item $\displaystyle r_{n,2}=0$
\end{enumerate} 
In other words, while the number of non-crossing parity-reversing pairings on $(\ast\circ)^n$ is counted by the $n^\text{th}$ Catalan number, there are no parity-reversing pairings containing exactly one or two crossings. 
\label{propnocross}
\end{proposition}

\begin{figure}[htbp]
\includegraphics[scale=0.5]{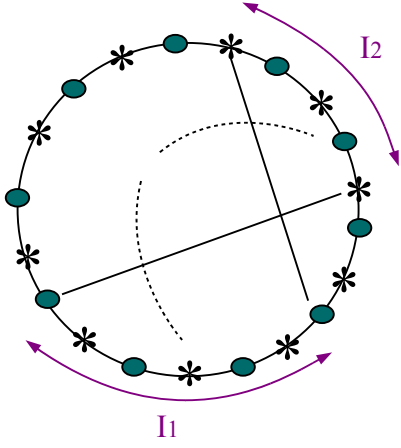}\centering
\caption{Proof of $r_{n,k}=b_{n,k}=0$ for $k=1,2$. The dotted lines represent the segments of chord(s) that pair an element of $I_1$ to an element in its complement, and similarly for $I_2$.}
\label{proof2}
\end{figure}

\begin{proof}
It is well known (e.g. \cite{Errera1931}) that the non-crossing pairings on $[2n]$ are in bijective correspondence with the non-crossing partitions on $[n]$, from which it follows that the non-crossing pairings on $\{1,\ldots,2n\}$ are counted by the $n^\text{th}$ Catalan number $C_n={{2n}\choose n}/(n+1)$. Observing that any non-crossing pairing on $\{1,\ldots,2n\}$ must pair even to odd integers only (i.e. it must be parity-reversing), part (1) follows. 

Next let $(\circ_i,\ast_j)$ and $(\circ_k,\ast_\ell)$ be two pairs that cross, where $1\leq i<j\leq n$. Consider the two possible intervals delineated by $\circ_i$ and $\circ_j$, one of length $2(j-i)+1$ and the other of length $2n-2(j-i)+1$ including the end points. Since the pairs cross,  $\ast_j$ and $\ast_\ell$ must belong to the same interval. Denote by $I_1$ the interval between $\circ_i$ and $\circ_j$ that does not contain $\ast_j$ and $\ast_\ell$. Similarly, let $I_2$ be the interval delimited by $\ast_j$ and $\ast_\ell$ that does not contain $\circ_i$ and $\circ_j$. The two intervals are depicted in Figure~\ref{proof2}.
Since the interval $I_1$ contains an odd number of unpaired elements, the pairing in question must contain a chord that connects an unpaired element of $I_1$ to an element in the complement of $I_1$, thus inducing an additional crossing and proving claim (2). But, in fact, since $I_1$ and $I_2$ are disjoint and $I_2$ also contains an odd number of elements, there must be not one but two additional crossings, proving claim (3).
\end{proof}

\begin{remark} Proposition~\ref{propnocross} was observed by Tony Wong, who worked on enumeration of alternating bitstrings during the 2007 Research Experiences for Undergraduates (REU) program.
\end{remark}

\begin{prop}
For $n\geq 1$ and $k\geq 3$,
\begin{enumerate}
\item For $n$ odd, any parity-reversing pairing on $(\ast\circ)^n$ achieves at most ${n\choose 2}$ crossings. The unique pairing that achieves ${n\choose 2}$ crossings is connected and is given by $\{(i,n/2+(2i-1)/i)\}_{i=1,\ldots,n}$. Thus, $r_{n,{n\choose 2}}=b_{n,{n\choose 2}}=1$ and $r_{n,k}=b_{n,k}=0$ for $k>{n\choose 2}$.
\item For $n$ even, any parity-reversing pairing on $(\ast\circ)^n$ achieves at most $n(n-2)/2$ crossings. The two pairings that achieve $n(n-2)/2$ crossings are connected and are given by $\{(i,n/2+i)\}_{i=1,\ldots,n}$  and $\{(i,n/2+i-1)\}_{i=1,\ldots,n}$. Thus, $r_{n,n(n-2)/2}=b_{n,n(n-2)/2}=1$ and $r_{n,k}=b_{n,k}=0$ for $k>n(n-2)/2$.
\item There exists no \emph{connected} parity-reversing pairing on $(\ast\circ)^n$ with less than $n$ crossings. For $n=3$, the unique pairing with three crossings is given by $\{(1,2),(2,3),(3,1)\}$. For $n\geq 4$, there are two pairings with $n$ crossings, given by $\{(i,i+1)\}_{i=1,\ldots,n}$ and $\{(i,i+2)\}_{i=1,\ldots,n}$.
In other words, $b_{3,3}=1$, $b_{n,n}=2$ for $n\geq 4$, and $b_{n,k}=0$ for $k<n$.
\end{enumerate}\label{prop pairings}
\end{prop}

\begin{proof} \begin{enumerate}
\item For any $n$, the upper bound follows immediately from the fact that there are ${n\choose 2}$ unordered pairs of chords in a diagram with $n$ chords. To exhibit the pairing that achieves the bound for $n$ odd, note that any chord partitions the circle into two halves: one containing $2m$ unpaired elements and the other $2(n-1-m)$ unpaired elements, for some $m=0,\ldots,n$. For the chord to intersect $n-1$ other chords, it is necessary that $m=(n-1)/2$. Since this is true for every chord, a unique pairing is obtained. The resulting pairing is illustrated in Fig.~\ref{max pairings}-a for $n=5$.

\item Suppose that some pairing achieves ${n\choose 2}$ crossings for $n$ even. It follows that each chord intersects $n-1$ other chords. Fix a chord and again note that the remainder of the circle is partitioned into two halves, containing $2m$ and $2(n-1-m)$ unpaired elements respectively. For the chord to intersect $n-1$ other chords, there must be $n-1$ chords starting in one half and ending in the other. This is impossible if $n$ is even as there exists no integer $m$ such that $2m=2(n-1-m)$.
 
Again fixing a chord, the maximal number of crossings involving that chord is in fact obtained when either $m=n/2$ (and therefore $(n-1-m)=n/2-1$) or $m=n/2-1$ (and therefore $(n-1-m)=n/2$). The two choices yield two pairings which are mirror images of eachother. For each pair, there are $n/2-1$ crossings, and the total number of crossings is therefore $n(n-2)/2$. The corresponding pairings are illustrated in Fig.~\ref{max pairings}-b for $n=4$.

\begin{figure}[tbp]
\includegraphics[scale=0.5]{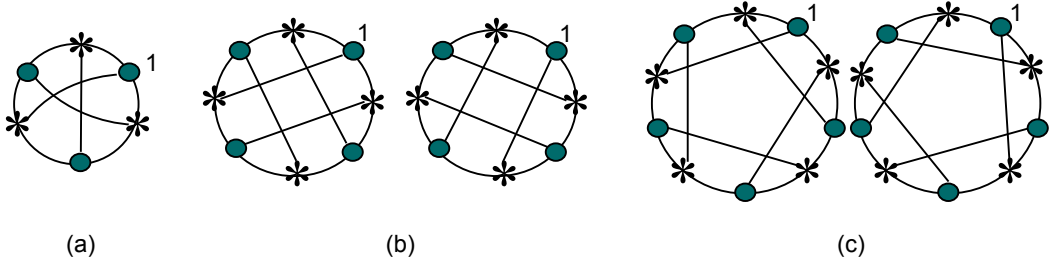}\centering
\caption{Connected parity-reversing pairings with the minimal number of crossings for (a) $n=3$, (b) $n=4$, (c) $n=5$.}
\label{min pairings}
\end{figure}

\begin{figure}[tbp]
\includegraphics[scale=0.5]{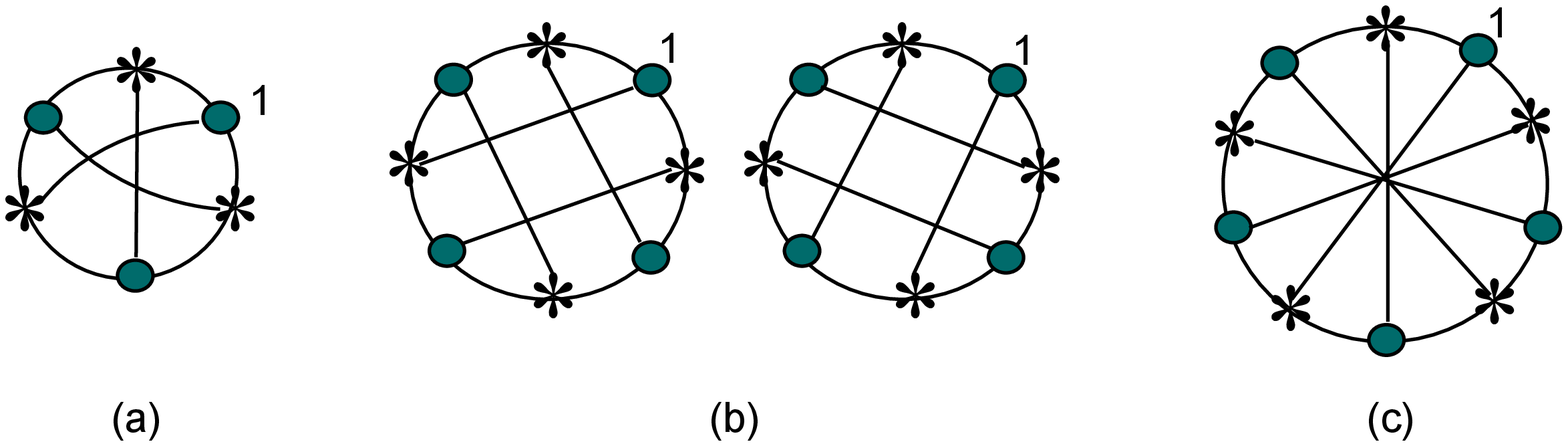}\centering
\caption{Parity-reversing pairings with the maximal number of crossings for (a) $n=3$, (b) $n=4$, (c) $n=5$. Note that a parity-reversing pairing with a maximal number of crossings is necessarily connected.}
\label{max pairings}
\end{figure}

\item A necessary condition for a pairing to be connected is that every chord must intersect at least one other chord. It follows that in a diagram with $n$ chords, $b_{n,k}=0$ if $k<n$. Noting that ${n\choose 2}=3$ for $n=3$, that $b_{3,3}=1$ follows by (2). Due to the parity-reversing structure, if a chord intersects another, it must intersect at least two chords (see Figure~\ref{proof2} and the proof of Proposition~\ref{propnocross}). Since a single crossing is shared by exactly two chords, fixing the total number of crossings to be $n$ yields that every chord can intersect at most two chords. The only two possibilities are therefore $\sigma(i)=i+1$ for $i=1,\ldots,n$ and $\sigma(i)=i-2$ for $i=1,\ldots,n$, where the indices should be interpreted modulo $n$. For $n=3$, the two permutations coincide, while the two possibilities for $n=4,5$ are illustrated in Fig.~\ref{min pairings}.
\end{enumerate}
\end{proof}

\begin{remark}Recalling our original purpose, the norm of the $q$-circular operator can be expressed via Propositions~\ref{propnocross} and \ref{prop pairings} as
\[||c_q||=\lim_{n\to\infty}\left(\frac{1}{n+1}{{2n}\choose{n}}+\sum_{k=3}^{n\choose 2} r_{n,k}q^k\right)^\frac{1}{2n}.\]
Letting $q=0$ recovers the familiar result of free probability, namely that $||c_0||=2$, where $c_0$ denotes the circular operator realized as a sum of two free copies of the semi-circular operator. (Specifically, $c_0=(s_0+i\tilde s_0)/\sqrt 2$ where $s_0$ and $\tilde s_0$ are two free semi-circular operators normalized so that $\phi(s_0^2)=\phi(\tilde s_0^2) = 1$.)\label{remarknormq0}
\end{remark}

The combinatorial approach used to characterize the analytic behavior of $||c_q||$ as a function of $q$ is based on the decomposition of a pairing into connected components and the resulting expression for $r_{n,k}$ in terms of finite products of the $\{b_{\ell,m}\}_{\ell,m}$ sequence. Consider first the following reformulation of (\ref{decompositionEq2}).

\begin{lem} For all $n\geq 1$, $k\geq 0$,
\begin{equation}r_{n,k}=\!\!\!\!\!\!\!\!\!\sum_{\substack{\ell=1,\ldots,k\\\beta=\{(n_1,k_1),\ldots,(n_\ell,k_\ell)\} \\ n_1+\ldots+n_\ell= n\\ k_1+\ldots+k_\ell=k\\n_i,k_i\in Z_+}}\frac{(2n)!}{\Phi_1(\beta)!\Phi_2(\beta)!\ldots \Phi_n(\beta)!(2n+1-l)!}b_{n_1,k_1}\ldots b_{n_\ell,k_{\ell}},\label{decomposition}\end{equation}
where $\Phi_i(\beta)$ counts the number of pairs $\beta$ with the first coordinate equaling $i$, that is, $$\Phi_i(\beta)=\Big|\{(\tilde n,\tilde k)\in\beta\mid \tilde n=i\}\Big|.$$
\label{P}
\end{lem}
\begin{proof}
Based on the decomposition based on non-crossing partitions introduced in Section~\ref{connectedSec}, given any parity-reversing pairing on $(\ast\circ)^n$ with $k$ crossings, the pairing decomposes uniquely into $\ell$ connected components, for some  positive integer $\ell\leq n$, of sizes $n_1,\ldots,n_\ell$ with $k_1,\ldots,k_\ell$ crossings respectively. An expression for $r_{n,k}$ is therefore obtained by considering all choices of $\ell\in\{1,\ldots,n\}$, all pairs $(n_1,k_1),\ldots,(n_\ell,k_\ell)\in \N\times \N$ with $n_1+\ldots+n_\ell= n$ and $k_1+\ldots+k_\ell=k$, all choices of the corresponding connected components, and, finally, all the ways of ``assembling'' a parity-reversing pairing on $(\ast\circ)^n$ from the corresponding connected components. The choice of connected components is realized as a product of $b_{n_1,k_1}\ldots b_{n_\ell,k_\ell}$, while the multiplicity associated with arranging the components into a larger pairing is that of partitioning the set $[n]$ into non-crossing parts of sizes $n_1,\ldots,n_\ell$. It is well known (e.g. Corollary 9.13 in \cite{NicaSpeicher}) that for any positive integer $n$ and any  $r_1, \ldots , r_n\in \N \cup \{0\}$ such that $r_1 + 2r_2 + \ldots + nr_n = n$, the number of partitions $\pi\in NC(n)$ that have $r_1$ 1-blocks, $r_2$ 2-blocks, \ldots, $r_n$ $n$-blocks is given by
$$\frac{n!}{r_1!r_2!\ldots r_n!(n+1-(r_1+r_2+\ldots+r_n))!}.$$
Counting the multiplicities of each $1$-block, $2$-block, and so on, among $n_1,\ldots,n_\ell$ yields the multiplier in (\ref{decomposition}).
\end{proof}

\vspace{10pt}

The combinatorial insight derived from Proposition~\ref{prop pairings} and enabling the characterization of $||c_q||$ via the above decomposition hinges on the fact that for fixed $k\in \N$, $b_{n,k}=0$ for all $n>k$. Therefore, for all $n$, the expression (\ref{decomposition}) for $r_{n,k}$ is a sum of products of multinomial factors in $n$ and coefficients $b_{m,\ell}$, where $m\leq \ell\leq k$. In particular, for fixed $k$, the subscripts $m$ and $\ell$ do not depend on $n$.
It follows that based on the knowledge of $b_{m,\ell}$ for $1\leq\ell \leq {k\choose 2}$ and $0\leq m\leq k$, obtained via an oracle or direct enumeration, entirely determines the values of $r_{n,k}$ \emph{for all $n$}.

Table~\ref{table B} displays the non-zero values of $b_{n,k}$ for $k\leq 11$ computed by direct enumeration.\footnote{Since parity-reversing pairings on $(\ast\circ)^n$ are in bijective correspondence with the permutations on $[n]$, Table~\ref{table B} was generated in Matlab by listing the $n!$ permutations and counting the crossings as in Remark~\ref{remarkcross}.} These values are used in the following proposition to obtain closed-form expressions for $r_{n,3},\,r_{n,4},\ldots,\,r_{n,11}$, which are in turn employed in Section~\ref{derivative} to derive several analytic properties of $||c_q||$. 

For example, to obtain that $r_{n,3}={{2n}\choose {n-3}}$, note that $3={3\choose 2}$ and thus $b_{n,3}=0$ for all $n\neq 3$. It follows that the only terms contributing to the sum in (\ref{decomposition}) are those where $b_{n_i,k_i}\in\{b_{1,0},b_{3,3}\}$ for all $i=1,\ldots,l$. But, since $k=3$, there is only one such contributing term, corresponding to $\ell=n-2$ (i.e. one factor $b_{3,3}$ and $n-3$ factors $b_{1,0}$). Its contribution is given by $\frac{(2n)!}{(n-3)!(2n+1-(n-2))!}={2n\choose n-3}$ and result then follows from Table~\ref{table B}, noticing that $b_{1,0}=b_{3,3}=1$.

\begin{table}[tbp]
\begin{tabular}{|r|c|c|c|c|c|c|c|c|c|c|c|}
\hline
&\textbf{n=1}&\textbf{2}&\textbf{3}&\textbf{4}&\textbf{5}&\textbf{6}&\textbf{7}&\textbf{8}&\textbf{9}&\textbf{10}&\textbf{11}\\\hline
\textbf{k=0}&1&0&0&0&0&0&0&0&0&0&0\\\hline
\textbf{k=3}&0&0&1&0&0&0&0&0&0&0&0\\\hline
\textbf{4}&0&0&0&2&0&0&0&0&0&0&0\\\hline
\textbf{5}&0&0&0&0&2&0&0&0&0&0&0\\\hline
\textbf{6}&0&0&0&0&5&2&0&0&0&0&0\\\hline
\textbf{7}&0&0&0&0&5&24&2&0&0&0&0\\\hline
\textbf{8}&0&0&0&0&0&18&56&2&0&0&0\\\hline
\textbf{9}&0&0&0&0&0&4&70&176&2&0&0\\\hline
\textbf{10}&0&0&0&0&1&12&98&328&576&2&0\\\hline
\textbf{11}&0&0&0&0&0&12&105&408&1107&300&2\\\hline
\end{tabular}

\caption{$b_{n,k}$ for $0\leq n,k\leq 11$. The omitted rows corresponding to $k=1,2$ are null by Proposition~\ref{propnocross}.}\label{table B}
\end{table}

\break

 \begin{lem} The following expressions, valid for $n\in\mathbb N$ under the convention that ${{2n}\choose {n-\ell}}=0$ for $\ell\geq n$,  count the number of parity-reversing pairings on $(\ast\circ)^n$ containing between $3$ and $11$ crossings, respectively.

{\footnotesize \begin{eqnarray*}\displaystyle r_{n,3}&=&{{2n}\choose {n-3}}\\
r_{n,4}&=&2{{2n}\choose {n-4}}\\
r_{n,5}&=&2{{2n}\choose {n-5}}\\
r_{n,6}&=&5{{2n}\choose {n-5}}+2{{2n}\choose {n-6}}+\frac{n+6}{2}{{2n}\choose {n-6}}\\
r_{n,7}&=&5{{2n}\choose {n-5}}+24{{2n}\choose {n-6}}+2{{2n}\choose {n-7}}+2(n+7){{2n}\choose {n-7}}\\
r_{n,8}&=&18{{2n}\choose {n-6}}+56{{2n}\choose {n-7}}+2{{2n}\choose {n-8}}+4(n+8){{2n}\choose {n-8}}\\
r_{n,9}&=&4{{2n}\choose {n-6}}+70{{2n}\choose {n-7}}+176{{2n}\choose {n-8}}+2{{2n}\choose {n-9}}+5(n+8){{2n}\choose {n-8}}+6(n+9){{2n}\choose {n-9}}\\&+&\frac{1}{6}(n+9)(n+8){{2n}\choose {n-9}}\\
r_{n,10}&=&{{2n}\choose {n-5}}+12{{2n}\choose {n-6}}+98{{2n}\choose {n-7}}+328{{2n}\choose {n-8}}+576{{2n}\choose {n-9}}+2{{2n}\choose {n-10}}+5(n+8){{2n}\choose {n-8}}\\&+&34(n+9){{2n}\choose {n-9}}+8(n+10){{2n}\choose {n-10}}+(n+10)(n+9){{2n}\choose {n-10}}\\
r_{n,11}&=&12{{2n}\choose {n-6}}+105{{2n}\choose {n-7}}+408{{2n}\choose {n-8}}+1107{{2n}\choose {n-9}}+300{{2n}\choose {n-10}}+2{{2n}\choose {n-11}}+28(n+9){{2n}\choose {n-9}}\\&+&109(n+10){{2n}\choose {n-10}}+10(n+11){{2n}\choose {n-11}}+3(n+11)(n+10){{2n}\choose {n-11}}\end{eqnarray*}}
\label{expressionsLem}
\end{lem}

\section{Analytic properties of $||c_q||$}
\label{derivative}

We now make use of the combinatorial results of the previous section to characterize the convergence of the $2n$-norms of the $q$-circular operator $c_q$.  Recalling the notation, let $\lambda_n:\R\to\R$ denote a real-valued function in $q$ given by
\begin{equation}\lambda_n(q)=\left(\sum_{k=0}^{n\choose 2} p_{n,k}q^k\right)^\frac{1}{2n},\label{momentsEq}\end{equation}
where $p_{n,k}$ again denotes the number of parity-reversing pairings on $(\ast\circ)^n$ with $k$ crossings. 
By Definition~\ref{defMoments}, for $q\in(-1,1)$, $\lambda_n(q)$ is the $2n$-norm of $c_q$ and by (\ref{normc}),  
\begin{equation}||c_q||=\lim_{n\to\infty}\lambda_n(q).\label{normcqEq}\end{equation}
Since for all $n\in \mathbb N$, $\lambda_n$ is (real) analytic on $(-1,1)$, it admits an analytic extension $\tilde \lambda_n$ onto some domain $\Omega$ in the complex plane. It is not clear whether the limit $||c_q||$, viewed as a function in $q$ on $(-1,1)$ is analytic, but, in the affirmative, its analytic extension onto a complex domain would allow the problem of the norm to be framed within a significantly richer theory.
A priori, one may hope that point-wise convergent analytic extensions of the sequence $\lambda_n$ will converge to the analytic extension of $||c_q||$. Surprisingly, this does not end up being the case and the underlying argument is the focus of the present section.

To rephrase Theorem~\ref{thmbad}, take the sequence of analytic extensions $\tilde \lambda_1,\tilde \lambda_2,\ldots$ of $\lambda_1,\lambda_2,\ldots$ and suppose that they converge point-wise on some fixed complex neighborhood of $(-1,1)$ to a function that we would hope is the analytic extension of $||c_q||$. Unfortunately, it will turn out that either the domains of analyticity of $\tilde \lambda_1,\tilde \lambda_2,\ldots$  will shrink to a single point or, otherwise, the convergence will fail to be uniform on any neighborhood of the origin. Either way, it is no longer clear whether the analytic extension of the operator norm of $c_q$ can be understood via the analytic extensions of its $2n$-norms.

More concretely, since
$$\tilde \lambda_n(q)=\left(\sum_{k=0}^{n\choose 2} p_{n,k}q^k\right)^\frac{1}{2n},\quad q\in\mathbb C,$$
the domain of analyticity of $\lambda_n(q)$ is delineated by the least-magnitude root of the polynomial $\sum_{k=0}^{n\choose 2} p_{n,k}q^k$. Unfortunately, visualizing the corresponding domains of analyticity requires the knowledge of $p_{n,k}$ for large $n$ and all $k=0,\ldots,{n\choose 2}$. Extending Table~\ref{table B} to $n=12$ and $k=0,\ldots,66$ allows one to compute the roots of $\tilde\lambda_3,\ldots,\tilde\lambda_{12}$. The least-magnitude root for each of the $\tilde\lambda_n$ is depicted in Figure~\ref{figRoots}. Though the roots appear to form a decreasing sequence, the data points are relatively few and it is therefore not clear which of two the behaviors identified in Theorem~\ref{thmbad} is actually taking place.

\begin{figure}[tbph]
\includegraphics[scale=1]{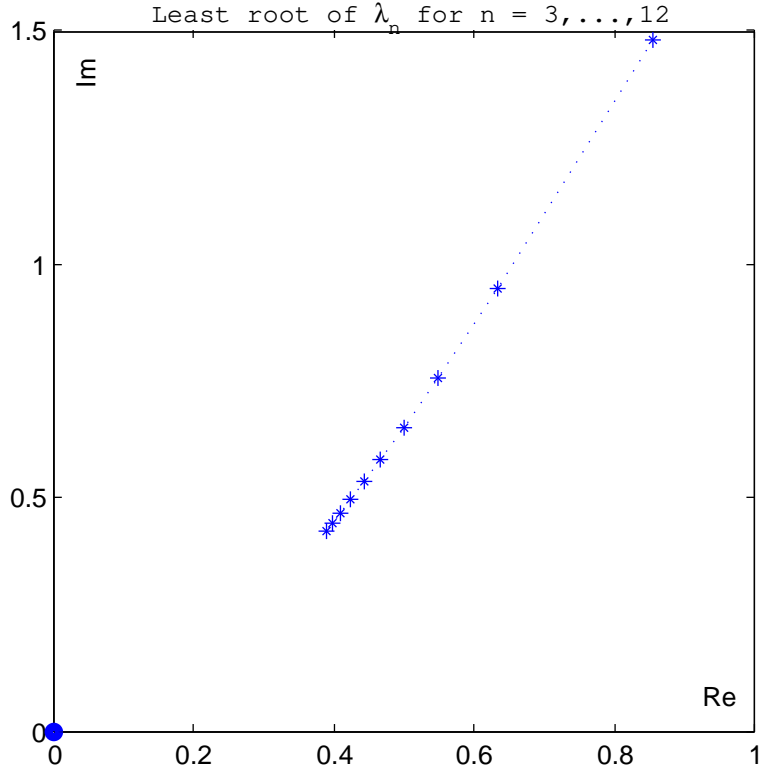}\centering
\caption{The least-magnitude roots of $\tilde \lambda_3,\ldots,\tilde \lambda_{12}$ (the conjugate roots in the fourth quadrant not displayed). It is to be noted that the magnitude of roots in this range decreases monotonically.}
\label{figRoots}
\end{figure}

The proof of Theorem~\ref{thmbad} hinges on the combinatorics of parity-reversing pairings developed in the previous section in conjuction with the Taylor series expansions of $\lambda_n$ developed in the following Lemmas~\ref{PS} and \ref{lemma11}.
 
\begin{lem} For all $n\geq 1$, the function $\lambda_n$ is infinitely differentiable on $(-1,1)$. Moreover, letting $\{a_k^{(n)}\}_{k\geq 0}$ denote the coefficients in the Taylor series expansion of $\lambda_n$ about the origin, we have
$$a_0=(p_{n,0})^\frac{1}{2n}$$ and for $k\geq 1$,
\begin{eqnarray}
a_k^{(n)}&=&(p_{n,0})^{\frac{1}{2n}}\sum_{\substack{\ell_1,\ell_2,\ldots,\ell_k\in \{0,1,\ldots,k\}\\\ell_1+2\ell_2+\ldots+k\ell_k=k}}\!\!\!\frac{\frac{1}{2n}!}{(\frac{1}{2n}-(\ell_1+\ldots+\ell_k))!}\frac{1}{\ell_1!\ldots \ell_k!}\left(\frac{p_{n,1}}{p_{n,0}}\right)^{\ell_1}\ldots \left(\frac{p_{n,k}}{p_{n,0}}\right)^{\ell_k}.\label{thm2}
\end{eqnarray}
\label{PS}
\end{lem}
\begin{proof}
Recall first that letting $f,g:\R\to\R$ be two functions such that $D^k f, D^kg$ exist for all $k=1,\ldots,n$ given some $n\in Z_+$, an extension of the chain rule to the $k^\text{th}$ derivative of $f\circ g$ (e.g. \cite{Stanley1997v2}, Chapter 5) can be written as follows
{\footnotesize$$D^n (g\circ f)
=\sum_{\substack{m_1,m_2,\ldots,m_n\in \{0,1,\ldots,n\}\\m_1+2m_2+\ldots+nm_n=n}} \frac{n!}{m_1!\,m_2!\,\cdots\,m_n!}\;
(D^{m_1+\cdots+m_n}g) \circ f\;\times\;
\prod_{j=1}^n\left(\frac{D^jf}{j!}\right)^{m_j}.$$}
Specializing the above result, suppose that $f:M\to\R$ is a smooth function that does not vanish on some open subset $M\subset \R$ and let $g(x)=x^b$. Then, for any $b\in \R$, $f^b$ is smooth.  Moreover, for $b\neq 0$,
$$(D^{m_1+\cdots+m_n}g)\circ f= \frac{b!}{(b-(\ell_1+\ldots+\ell_k))!} f^{-(\ell_1+\ldots+\ell_k)}$$
and, therefore,
{\footnotesize$$D^k f^b=f^b\sum_{\substack{\ell_1,\ell_2,\ldots,\ell_k\in \{0,1,\ldots,k\}\\\ell_1+2\ell_2+\ldots+k\ell_k=k}}\!\!\!\frac{k!}{\ell_1!\ldots \ell_k!}\frac{b!}{(b-(\ell_1+\ldots+\ell_k))!}\left(\frac{D^1f}{1!f}\right)^{\ell_1}\;\left(\frac{D^2f}{2!f}\right)^{\ell_2}\;\ldots \left(\frac{D^kf}{k!f}\right)^{\ell_k},$$}
for all $k\in\Z_+$, where $\frac{b!}{(b-(\ell_1+\ldots+\ell_k))!}:=b\,(b_1)\ldots(b+1-(\ell_1+\ldots+\ell_k))$.

Recall now that given a $C^\ast$ probability space $(\mathcal A,\phi)$ of Definition~\ref{Cast}, for any non-zero element $a\in\mathcal A$, $\phi((a^\ast a)^n)>0$ for all $n\geq 1$. Since $\lambda_n(q)=(\phi((c_q^\ast c_q)^n))^{1/(2n)}$ and since $c_q$ can be represented an element of some $C^\ast$ probability space for all $q\in(-1,1)$, it follows that $\lambda_n>0$ on $(-1,1)$. The previous discussion yields that $\lambda_n$ is smooth and {\footnotesize $$
D^k \lambda_n =\sum_{\substack{\ell_1,\ell_2,\ldots,\ell_k\in \{0,1,\ldots,k\}\\\ell_1+2\ell_2+\ldots+k\ell_k=k}}\!\!\!\frac{k!}{\ell_1!\ldots \ell_k!}\frac{\frac{1}{2n}!}{(\frac{1}{2n}-(\ell_1+\ldots+\ell_k))!}(\lambda_n)^{\frac{1}{2n}}\;\left(\frac{D^1\lambda_n}{1!\lambda_n}\right)^{\ell_1}\ldots \left(\frac{D^k\lambda_n}{k!\lambda_n}\right)^{\ell_k}$$}
for all $k\in\Z_+$. But note that for any $m\in Z_+$, we have $D^{m}\lambda_n(0)=m!\,p_{n,m}$. Thus,
{\footnotesize\begin{eqnarray*}a_k^{(n)}&=&\frac{D^k \lambda_n(0)}{k!}\\&=&(p_{n,0})^{\frac{1}{2n}}\sum_{\substack{\ell_1,\ell_2,\ldots,\ell_k\in \{0,1,\ldots,k\}\\\ell_1+2\ell_2+\ldots+k\ell_k=k}}\!\!\!\frac{1}{\ell_1!\ldots \ell_k!}\frac{\frac{1}{2n}!}{(\frac{1}{2n}-(\ell_1+\ldots+\ell_k))!}\left(\frac{p_{n,1}}{p_{n,0}}\right)^{\ell_1}\ldots \left(\frac{p_{n,k}}{p_{n,0}}\right)^{\ell_k}.\end{eqnarray*}}
\end{proof}

More concretely, below is the expansion of the first ten coefficients of the Taylor series associated \mbox{with $\lambda_n$}. Combining this result with the expressions of Lemma~\ref{expressionsLem} for the combinatorial sequences $p_{n,k}$ yields the surprising asymptotics of Lemma~\ref{lemma11}.

{\footnotesize \begin{eqnarray}
\lambda_n(q)&=&(p_{n,0})^\frac{1}{2n}+(p_{n,0})^\frac{1}{2n}\frac{1}{2}\,\frac{p_{n,3}}{np_{n,0}}q^3 + (p_{n,0})^\frac{1}{2n}\frac{1}{2}\,\frac{p_{n,4}}{np_{n,0}}q^4
+(p_{n,0})^\frac{1}{2n}\frac{1}{2}\,\frac{p_{n,5}}{np_{n,0}}q^5\\\nonumber & +&  (p_{n,0})^{1\over{2n}}\left({1\over 2}{p_{n,6}\over{np_{n,0}}}+\left({1\over{2n}}-1\right){1\over 4}{(p_{n,3})^2\over{n(p_{n,0})^2}}\right)q^6
+(p_{n,0})^\frac{1}{2n}\left( \frac{1}{2}\,\frac {p_{n,7}}{np_{n,0}}+\left({1\over{2n}}-1\right)\frac{1}{2}\,\frac {p_{n,3}p_{n,4}}{n(p_{n,0})^2}\right){q}^{7}\\\nonumber &+& (p_{n,0})^\frac{1}{2n}\left(\frac{1}{2}\,\frac {p_{n,8}}{np_{n,0}}+ \left({1\over{2n}}-1\right)\frac{1}{4}\,\frac {(p_{n,4})^2}{n(p_{n,0})^2}+\right.
+\left. \left({1\over{2n}}-1\right)\frac{1}{2}\,\frac {p_{n,3}p_{n,5}}{n(p_{n,0})^2}\right){q}^{8}\\\nonumber
&+&(p_{n,0})^\frac{1}{2n}\left(\frac{1}{2}\,\frac {p_{n,9}}{np_{n,0}}+ \left({1\over{2n}}-1\right)\frac{1}{2}\,\frac {p_{n,4}p_{n,5}}{n(p_{n,0})^2}+\left({1\over{2n}}-1\right)\frac{1}{2}\,\frac {p_{n,3}p_{n,6}}{n(p_{n,0})^2}+\left({1\over{2n}}-2\right)\left({1\over{2n}}-1\right)\frac{1}{12}\,\frac {(p_{n,3})^3}{n(p_{n,0})^3}\right)q^{9}\\\nonumber
&+& O(q^{10})
\label{maple expansion}
\end{eqnarray}}

\begin{lem} For $k\in\{0,\ldots,10\}$, the sequence $a_{k}^{(n)}$ converges to a limit in $\R$ as $n\to\infty$. However, $$a_{11}^{(n)}\sim -5 n.$$\label{lemma11}\end{lem}
\begin{proof}
By substituting the expressions of Lemma~\ref{expressionsLem} into Lemma~\ref{PS}, the reader may readily verify that the sequence $a_{k}^{(n)}$ converges to a limit for $k\in\{0,\ldots,10\}$. In fact, using Lemma~\ref{P} rather than Lemma~\ref{expressionsLem}, the corresponding limits turn out to be weighted products of $b_{\ell,m}$ terms. For instance,
$$a_{3}^{(n)}=\frac{b_{3,3}}{2n}{{2n}\choose {n-3}}\left(\frac{1}{n+1}{{2n}\choose n}\right)^{\frac{1}{2n}-1}.$$
Noting that 
$$\left(\frac{1}{n+1}{{2n}\choose n}\right)^{\frac{1}{2n}}\to 2\quad\quad\text{and}\quad\quad {{2n}\choose {n-3}}{{2n}\choose n}^{-1}\to 1,$$
one obtains that $a_{3}^{(n)}\to b_{3,3}=1$.

On the other hand, by Lemma~\ref{PS},
{\footnotesize \begin{eqnarray}
a_{11}^{(n)} &=& \left({\frac {{2\,n\choose n}}{n+1}}\right)^{\frac{1}{2n}}
 \frac{n+1}{n} \frac{1}{16 n^2{2\,n\choose n} ^{3}}\, \left( 8\,p_{n,{11}}{n}^{2} {2\,n\choose n}^{2} + 4\,p_{n,{8}}p_{n,{3}}n{2\,n\choose n}+4\,p_{n,{4
}}p_{n,{7}}n{2\,n\choose n}\right.\\\nonumber&+&4\,p_{n,{5}}p_{n,{6}}n{2\,n\choose n}-4\,p
_{n,{8}}p_{n,{3}}{n}^{2}{2\,n\choose n}-8\,p_{n,{8}}p_{n,{3}}{n}^{3}{2\,n
\choose n}-4\,p_{n,{4}}p_{n,{7}}{n}^{2}{2\,n\choose n}-8\,p_{n,{4}}p_{n,{7}}{n
}^{3}{2\,n\choose n}\\\nonumber&-&4\,p_{n,{5}}p_{n,{6}}{n}^{2}{2\,n\choose n}-8\,p_{n,{5}
}p_{n,{6}}{n}^{3}{2\,n\choose n}-3\,{p_{n,{3}}}^{2}p_{n,{5}}{n}^{2}+10\,{p_{n,
{3}}}^{2}p_{n,{5}}{n}^{3}+8\,{p_{n,{3}}}^{2}p_{n,{5}}{n}^{4}\\\nonumber&-&4\,{p_{n,{3}}}^{2
}p_{n,{5}}n-3\,p_{n,{3}}{n}^{2}{p_{n,{4}}}^{2}+10\,p_{n,{3}}{p_{n,{4}}}^{2}{n}^{
3}+8\,p_{n,{3}}{p_{n,{4}}}^{2}{n}^{4}-4\,p_{n,{3}}n{p_{n,{4}}}^{2}+p_{n,{3}}{p_{n,
{4}}}^{2}+{p_{n,{3}}}^
{2}p_{n,{5}}\huge) \label{a11}
\end{eqnarray}}
To compute the asymptotic of the above expression, recall Stirling's formula,
$$n! \sim \sqrt{2 \pi n} \left(\frac{n}{e}\right)^n,$$
from which one readily obtains that for any \emph{fixed} non-negative integer $k$,
$${2\,n\choose n-k}\sim \frac{4^n}{\sqrt{\pi n}}.$$
Revisiting the expressions of Proposition~\ref{expressionsLem}, one obtains that for $3\leq k\leq 11$, 
$r_{n,k}$ is asymptotically given as a product of some polynomial in $n$ and a term that asymptotically equals $4^n/\sqrt{\pi n}$. 
For instance, $r_{n,11}=(3n^2+210n)(4^n/\sqrt{\pi n})(1+o(1))$. But, note that $r_{n,11}$ appears in (\ref{a11}) accompanied by $8n^2{2\,n\choose n}^{2}$, yielding the total contribution of $8n^2(3n^2+210n)(4^n/\sqrt{\pi n})^3(1+o(1))$. In fact, considering the asymptotics of the numerator in (\ref{a11}), it is easy to check that the additional binomial coefficients appear so that each term of the sum contributes an exponential term of $\left(\frac{4^n}{\sqrt{\pi n}}\right)^3$. That is,
$$\left( 8\,p_{n,{11}}{n}^{2} {2\,n\choose n}^{2} + 4\,p_{n,{8}}p_{n,{3}}n{2\,n\choose n}+\ldots+p_{n,{3}}{p_{n,{4}}}^{2}+{p_{n,{3}}}^{2}p_{n,{5}}\right)\sim P(n)\left(\frac{4^n}{\sqrt{\pi n}}\right)^3,$$
where $P(n)$ is some polynomial in $n$ (to be determined next). Thus,
\begin{equation}a_{11}^{(n)} \sim \frac{1}{8 n^2} P(n).\label{eqasymp}\end{equation}
The contribution to $P(n)$ from $8n^2 p_{n,11}$ is therefore $8n^2(3n^2+210n+O(1))$. Similarly, the contribution to $P(n)$ from $4\,p_{n,{8}}p_{n,{3}}n$ is $4n(4n+O(1))$, and so on. Computing the remaining contributions and performing the substitutions yields $P(n)= -40n^3 + O(n^2)$ (note that the term in $n^4$ vanished in the cancelations, but the term in $n^3$ did not, as discussed in the following Remark~\ref{remark11}). From (\ref{eqasymp}), one then obtains that $a_{11}^{(n)}\sim -5 n$, as claimed.
\end{proof}

\begin{remark} Rather than considering the asymptotic of the individual terms, performing the substitution of Lemma~\ref{expressionsLem} into Lemma~\ref{PS} by Maple yields
$$a_{11}^{(n)}=\left(\frac{1}{n+1}{{2n}\choose n}\right)^{\frac{1}{2n}}\frac{ - 20 n^{17} + 69 n^{16} + 10376 n^{15}   + O(n^{14})}{8 n^{16}+ 672 n^{15} + O(n^{14})}.$$
The claim of Lemma~\ref{lemma11} then follows directly as the large-$n$ asymptotic of the above expression.
\end{remark}

\begin{remark} Combining Lemma~\ref{P} and Lemma~\ref{PS} yields a general, albeit unweildy, expression for $a_k^{(n)}$ in terms of multinomial coefficients and products of the type $\prod_{i\leq n} b_{\ell_i,m_i}$ where $\sum \ell_i=\sum m_i=k$ for $\ell_i,m_i\geq 0$. The combinatorial reason for fact that $\lim_{n\to\infty} |a^{(n)}_{11}|=\infty$ whereas the limits of $a^{(n)}_{k}$ for $k<11$ are finite is the fact that $11$ is the least value of $k$ for which there exist at least two pairings on $[2k]$ with $k$ crossings that can be decomposed into irreducible pairings that match in the number of crossings, but mismatch in sizes. Specifically, in the $k=11$ case, a parity-reversing pairing on $[2k]$ with $k$ crossings can be constructed from three irreducible components, with the corresponding sizes and numbers of crossings given by $(n_1=5,k_1=5),(n_2=5,k_2=6)$, and $(n_3=1,k_3=0)$.\footnote{The numbers of ways to choose such irreducible pairings can be read off from Table~\ref{table B}.} At the same time, another valid pairing can be constructed with blocks given by $(\tilde n_1=5,\tilde k_1=5),(\tilde n_2=6,\tilde k_2=6)$, and $(\tilde n_3=0,\tilde k_3=0)$. The key observation at this point is that $k_1=\tilde k_1$, $k_2=\tilde k_2$, $k_3=\tilde k_3$, but $n_2\neq \tilde n_2$ and  $n_3\neq \tilde n_3$. Due to this mismatch, the combinatorial machine of Lemma~\ref{P} which counts the multiplicities of the block sizes and that of Lemma~\ref{PS} which counts the multiplicities of crossings will produce different numbers. In particular, the cancellations that make the $a_k^{(n)}$ limits finite for $k\leq 10$ will not occur.\label{remark11}
\end{remark}

The proof of Theorem~\ref{thmbad} now follows as a corollary of the above lemma.

\vspace{5pt}

\noindent{\it Proof of Theorem~\ref{thmbad}} Suppose that the first claim is false, viz. that there exists some open set $\Omega$ containing the origin on which $\tilde \lambda_n$ is analytic for all $n$. Without loss of generality, $\tilde\lambda_n$ converges point-wise on $\Omega$. By a 1901 result of Osgood \cite{Osgood1901}, there exists some dense open subset $O\subset\Omega$ on whose compact subsets the sequence $\tilde\lambda_n$ converges uniformly. Suppose that the origin belongs to $O$ and note that since $O$ is open, there exists some compact neighborhood of the origin on which $\tilde\lambda_n$ converges uniformly. A classical result from complex analysis (e.g. \cite{RudinRCA}) ensures that if a sequence of analytic functions on an open set in the complex plane converges uniformly on the compact subsets of the set, then for any positive integer $m$ the sequence of analytic functions on $\Omega$ formed by the $m$-fold derivatives of the original sequence converges uniformly on compact subsets of the open set. Thus, if $0\in O$, the eleventh derivative $\tilde\lambda^{(11)}_n$ of $\tilde\lambda_n$ converges uniformly on a compact neighborhood of the origin. But, by Lemma~\ref{lemma11}, this is impossible.
$\hfill\qed$.

\bigskip

\noindent {\bf Acknowledgments.} The author wishes to thank Todd Kemp for suggesting this problem, as well as for his advice and interesting conversations along the way. This paper was written during the author's stay at the University of California in San Diego.

\bibliographystyle{alpha}
\bibliography{q-circular-analyticity}

\begin{thebibliography}{Tou50b}

\bibitem[BBD06]{Biedl2006}
Therese Biedl, Franz Brandenburg, and Xiaotie Deng.
\newblock Crossings and permutations.
\newblock In Patrick Healy and Nikola Nikolov, editors, {\em Graph Drawing},
  volume 3843 of {\em Lecture Notes in Computer Science}, pages 1--12. Springer
  Berlin / Heidelberg, 2006.

\bibitem[Bia97]{Biane1997b}
Philippe Biane.
\newblock Free hypercontractivity.
\newblock {\em Comm. Math. Phys.}, 184(2):457--474, 1997.

\bibitem[BKS97]{Bozejko1997}
Marek Bo{\.z}ejko, Burkhard K{\"u}mmerer, and Rolandb Speicher.
\newblock {$q$}-{G}aussian processes: non-commutative and classical aspects.
\newblock {\em Comm. Math. Phys.}, 185(1):129--154, 1997.

\bibitem[BS91]{Bozejko1991}
Marek Bo{\.z}ejko and Roland Speicher.
\newblock An example of a generalized {B}rownian motion.
\newblock {\em Comm. Math. Phys.}, 137(3):519--531, 1991.

\bibitem[Cor07]{Corteel2007}
Sylvie Corteel.
\newblock Crossings and alignments of permutations.
\newblock {\em Advances in Applied Mathematics}, 38(2):149 -- 163, 2007.

\bibitem[CW10]{Corteel2010}
Sylvie Corteel and Lauren~K. Williams.
\newblock Staircase tableaux, the asymmetric exclusion process, and
  {A}skey-{W}ilson polynomials.
\newblock {\em Proc. Natl. Acad. Sci. USA}, 107(15):6726--6730, 2010.

\bibitem[Err31]{Errera1931}
Alfred Errera.
\newblock Analysis situs. un problème d'énumération.
\newblock {\em Académie royale de Belgique Mémoires de la Classe des
  Sciences}, 11(1421):26, 1931.
\newblock Fascicule 6.

\bibitem[FN00]{Flajolet2000}
Philippe Flajolet and Marc Noy.
\newblock Analytic combinatorics of chord diagrams.
\newblock {\em Technical Report}, n/a:191--201, 2000.

\bibitem[ISV87]{Ismail1987}
Mourad E.~H. Ismail, Dennis Stanton, and G{\'e}rard Viennot.
\newblock The combinatorics of {$q$}-{H}ermite polynomials and the
  {A}skey-{W}ilson integral.
\newblock {\em European J. Combin.}, 8(4):379--392, 1987.

\bibitem[Kem05]{Kemp2005}
Todd Kemp.
\newblock Hypercontractivity in non-commutative holomorphic spaces.
\newblock {\em Communications in Mathematical Physics}, 259:615--637, 2005.

\bibitem[MN01]{Mingo2001}
James Mingo and Alexandru Nica.
\newblock Random unitaries in non-commutative tori, and an asymptotic model for
  {$q$}-circular systems.
\newblock {\em Indiana Univ. Math. J.}, 50(2):953--987, 2001.

\bibitem[NS06]{NicaSpeicher}
Alexandru Nica and Roland Speicher.
\newblock {\em Lectures on the combinatorics of free probability}, volume 335
  of {\em London Mathematical Society Lecture Note Series}.
\newblock Cambridge University Press, Cambridge, 2006.

\bibitem[Osg02]{Osgood1901}
W.~F. Osgood.
\newblock Note on the functions defined by infinite series whose terms are
  analytic functions of a complex variable; with corresponding theorems for
  definite integrals.
\newblock {\em Ann. of Math. (2)}, 3(1-4):25--34, 1901/02.

\bibitem[Pen95]{Penaud1995}
Jean-Guy Penaud.
\newblock Une preuve bijective d'une formule de touchard-riordan.
\newblock {\em Discrete Mathematics}, 139(1-3):347 -- 360, 1995.

\bibitem[Rea79]{Read1979}
Ronald~C. Read.
\newblock The chord intersection problem.
\newblock {\em Ann. New York Acad. Sci.}, 319:444--454, 1979.

\bibitem[Rio75]{Riordan1975}
John Riordan.
\newblock The distribution of crossings of chords joining pairs of {$2n$}
  points on a circle.
\newblock {\em Math. Comp.}, 29:215--222, 1975.
\newblock Collection of articles dedicated to Derrick Henry Lehmer on the
  occasion of his seventieth birthday.

\bibitem[Rud87]{RudinRCA}
Walter Rudin.
\newblock {\em Real and complex analysis}.
\newblock McGraw-Hill Book Co., New York, third edition, 1987.

\bibitem[Sta99]{Stanley1997v2}
Richard~P. Stanley.
\newblock {\em Enumerative combinatorics. {V}ol. 2}, volume~62 of {\em
  Cambridge Studies in Advanced Mathematics}.
\newblock Cambridge University Press, Cambridge, 1999.
\newblock With a foreword by Gian-Carlo Rota and appendix 1 by Sergey Fomin.

\bibitem[Tou50a]{Touchard1950}
Jacques Touchard.
\newblock Contribution \`a l'\'etude du probl\`eme des timbres poste.
\newblock {\em Canadian J. Math.}, 2:385--398, 1950.

\bibitem[Tou50b]{Touchard1950-2}
Jacques Touchard.
\newblock Sur un probl\`eme de configurations.
\newblock {\em C. R. Acad. Sci. Paris}, 230:1997--1998, 1950.

\bibitem[Tou52]{Touchard1952}
Jacques Touchard.
\newblock Sur un probl\`eme de configurations et sur les fractions continues.
\newblock {\em Canadian J. Math.}, 4:2--25, 1952.

\bibitem[VDN92]{Voiculescu1992}
D.~V. Voiculescu, K.~J. Dykema, and A.~Nica.
\newblock {\em Free random variables}, volume~1 of {\em CRM Monograph Series}.
\newblock American Mathematical Society, Providence, RI, 1992.
\newblock A noncommutative probability approach to free products with
  applications to random matrices, operator algebras and harmonic analysis on
  free groups.

\bibitem[Wil05]{Williams2005}
Lauren~K. Williams.
\newblock Enumeration of totally positive {G}rassmann cells.
\newblock {\em Adv. Math.}, 190(2):319--342, 2005.

\end{thebibliography}

\end{document}